\documentclass[conference]{llncs}

\usepackage[cmex10]{amsmath}
\usepackage{amsfonts}
\usepackage{amssymb}

\newcommand{\pnt}[1]{{\mbox{\boldmath $#1$}}}

\newcommand{\so}[1]{$\{#1\}$}
\newcommand{\s}[1]{\mbox{$\{#1\}$}}

\newcommand{\prr}[1]{\mi{Prev}(\boldsymbol{q})}

\newcommand{\mi}[1]{\mathit{#1}}
\newcommand{\ti}[1]{\textit{#1}}
\newcommand{\tb}[1]{\textbf{#1}}

\newcommand{\Dss}[4]{\mbox{$(\prob{#1}{#2},\pnt{#3})~\rightarrow #4$}}

\newcommand{\ttt}{\>\>\>}

\newcommand{\Tt}{\>\>}
\newcommand{\Sub}[2]{\mbox{$\mi{#1}_\mi{#2}$}}
\newcommand{\Sup}[2]{\mbox{$\mi{#1}^\mi{#2}$}}
\newcommand{\SUp}[3]{\mbox{$\mi{#1}_{#2}^\mi{#3}$}}
\newcommand{\SUP}[3]{\mbox{$\mi{#1}_\mi{#2}^\mi{#3}$}}

\newcommand{\PRob}[2]{$\exists{#1} [#2]$}
\newcommand{\prob}[2]{\mbox{$\exists{#1} [#2]$}}
\newcommand{\Prob}[3]{\mbox{$\exists{#1}\exists{#2}[#3]$}}

\newcommand{\Comment}[1]{}

\newcommand{\Impl}[2]{\mbox{$#1 \rightarrow #2$}}
\newcommand{\Nmpl}[2]{\mbox{$#1 \not\rightarrow #2$}}
\newcommand{\EC}{$\mi{EC\!\_LoR}$~}

\newcommand{\ECM}{$\mi{EC\!\_LoR}^*$~}
\newcommand{\ecm}{$\mi{EC\!\_LoR}^*$}
\newcommand{\tl}[1]{\mbox{$\mi{TopLvl}(#1)$}}
\newcommand{\Pqe}[4]{\mbox{$\prob{#1}{#2 \wedge #4} \equiv #3 \wedge \prob{#1}{#4}$}}
\newcommand{\pqe}[4]{$\prob{#1}{#2 \wedge #4} \equiv #3 \wedge \prob{#1}{#4}$}

\newcommand{\al}{$\alpha$}
\newcommand{\Al}{$\alpha$~}
\newcommand{\be}{$\beta$}
\newcommand{\Be}{$\beta$~}
\newcommand{\m}[1]{$\mi{Mlp}_{#1}$}
\usepackage{wrapfig}
\usepackage{graphicx}
\usepackage{enumerate}
\begin{document}

\title{Equivalence Checking By Logic Relaxation}


\author{Eugene Goldberg}
\institute{\email{eu.goldberg@gmail.com}}

\maketitle

\begin{abstract}
We introduce a new framework for Equivalence Checking (EC) of Boolean
circuits based on a general technique called \tb{Lo}gic
\tb{R}elaxation (LoR).  The essence of LoR is to relax the formula to
be solved and compute a superset $S$ of the set of new behaviors.
Namely, $S$ contains all new satisfying assignments that appeared due
to relaxation and does not contain assignments satisfying the original
formula. Set $S$ is generated by a procedure called partial quantifier
elimination.  If all possible bad behaviors are in $S$, the original
formula cannot have them and so the property described by this formula
holds.  The appeal of EC by LoR is twofold.  First, it facilitates
generation of powerful \ti{inductive proofs}. Second, proving
inequivalence comes down to checking the presence of some bad behaviors in
the \ti{relaxed formula} i.e. in a simpler version of the original
formula. We give experimental evidence that supports our
approach.
\end{abstract}

\section{Introduction}

%
%
\subsection{Motivation}
Our motivation for this work is threefold.  First, \tb{Equivalence
  Checking (EC)} is a crucial part of hardware verification. Second,
more efficient EC enables more powerful logic synthesis
transformations and so strongly impacts design quality. Third,
intuitively, there should exist robust and efficient EC methods meant
for combinational circuits computing values in a ``similar
manner''. Once discovered, these methods can be extended to EC of
sequential circuits and even software.

\subsection{Structural similarity of circuits}
\label{ssec:ssim_circs}
In this paper, we target EC of structurally similar circuits $N'$ and
$N''$.  Providing a comprehensive definition of structural similarity
is a tall order.  Instead, below we give an example of circuits that
can be viewed as structurally similar.  Let $v'$ be a variable of
circuit $N'$.  Let $S(v') = \s{v''_{i_1},\dots,v''_{i_k}}$ be a set of
variables of $N''$ that have the following property.  The knowledge of
values assigned to the variables of $S(v')$ in $N''$ under input
\pnt{x} is sufficient to find out the value of $v'$ in $N'$ under the
same input \pnt{x}. (This means that $v'$ is a function of variables
of $S(v')$ modulo assignments to $S(v')$ that cannot be produced in
$N''$ under any input.) Suppose that for every variable $v'$ of $N'$
there exists a \ti{small} set $S(v')$ with the property above. Then
one can consider circuits $N'$ and $N''$ as \tb{structurally
  similar}. The smaller sets $S(v')$, the closer $N'$ and $N''$
structurally. In particular, if $N'$ and $N''$ are identical, for
every variable $v'$ there is a set $S(v')$ consisting only of one
variable of $N''$.  An example of structurally similar circuits where
$S(v')$ consists of two variables is given in
Subsection~\ref{subsec:ec_lor}.

%
%
\subsection[EC by logic relaxation]{EC by logic relaxation}
\label{ssec:ec_by_lor}
Let $N'(X',Y',z')$ and $N''(X'',Y'',z'')$ be single-output circuits to
be checked for equivalence. Here $X'$ and $Y'$ specify the sets of
input and internal variables of $N'$ respectively and $z'$ specifies
the output variable of $N'$. The same applies to $X'',Y'',z''$ of
circuit $N''$.  A traditional way to verify the equivalence of $N'$
and $N''$ is to form a two-output circuit shown in
Fig.~\ref{fig:trad_ec} and check if $z' \neq z''$ for some input
assignment (\pnt{x'},\pnt{x''}) where \pnt{x'}=\pnt{x''}. Here
\pnt{x'} and \pnt{x''} are assignments to variables of $X'$ and $X''$
respectively.  (By saying that \pnt{p} is an assignment to a set of
variables $V$, we will assume that \pnt{p} is a \ti{complete}
assignment unless otherwise stated. That is every variable of $V$ is
assigned a value in \pnt{p}.)

Formula $\mi{EQ}(X',X'')$ relating inputs of $N'$ and
$N''$ in Fig.~\ref{fig:trad_ec} evaluates to 1 for assignments
\pnt{x'} and \pnt{x''} to $X'$ and $X''$ iff \pnt{x'}=\pnt{x''}.
(Usually, $N'$ and $N''$ are just assumed to share the same set of
input variables. In this paper, for the sake of convenience, we let
$N'$ and $N''$ have separate sets of input variables but assume that
$N'$ and $N''$ must be equivalent only for the input assignments
satisfying $\mi{EQ}(X',X'')$.)

EC by \tb{Logic Relaxation (LoR)} presented in this paper is based on
the following idea.  Let $\mi{Out}(N',N'')$ denote the set of outputs
of $N'$ and $N''$ when their input variables are constrained by
$\mi{EQ}(X',X'')$.  For equivalent circuits $N'$ and $N''$ that are
not constants, $\mi{Out}(N',N'')$ is equal to
\so{(z'=0,z''=0),(z'=1,z''=1)}.  Let $\Sup{Out}{rlx}(N',N'')$ denote
the set of outputs of $N'$ and $N''$ when their inputs are not
constrained by $\mi{EQ}(X',X'')$ . (Here $\mi{rlx}$ stands for
``relaxed'').  $\Sup{Out}{rlx}(N',N'')$ is a superset of
$\mi{Out}(N',N'')$ that may contain an output $(z'=0,z''=1)$ and/or
output $(z'=1,z''=0)$ even when $N'$ and $N''$ are equivalent.  Let
\Sub{D}{out} denote $\Sup{Out}{rlx}(N',N'') \setminus
\mi{Out}(N',N'')$. That is \Sub{D}{out} contains the outputs that can
be produced \ti{only} when inputs of $N'$ and $N''$ are independent of
each other.

\setlength{\intextsep}{2pt}
\begin{wrapfigure}{l}{1.5in}
 \begin{center}
    \includegraphics[height=1.2in,width=1.5in]{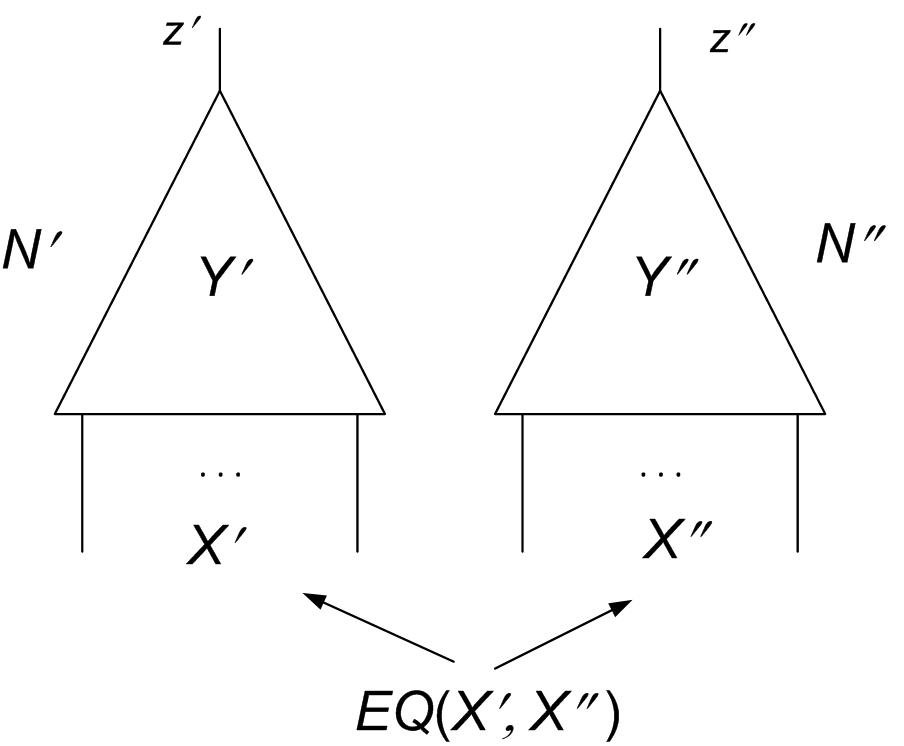}
  \end{center}
\vspace{-10pt}
\caption{Equivalence checking of $N'$ and $N''$}
\label{fig:trad_ec}
\end{wrapfigure}

Computing set \Sub{D}{out} either solves the equivalence checking of
$N'$ and $N''$ or dramatically simplifies it. (This is important
because, arguably, set \Sub{D}{out} is much easier to find than
$\mi{Out}(N',N'')$.) Indeed, assume that \Sub{D}{out} contains both
$(z'=1,z''=0)$ and $(z'=0,z''=1)$. Then $N'$ and $N''$ are
\ti{equivalent} because assignments where values of $z'$ and $z''$ are
different are present in $\Sup{Out}{rlx}(N',N'')$ but not in
$\mi{Out}(N',N'')$. Now, assume that \Sub{D}{out} does not contain,
say, assignment $(z'=1,z''=0)$. This can only occur in the following
two cases. First, $N'$ cannot produce output 1 (i.e. $N'$ is a
constant 0) and/or $N''$ cannot produce 0 (i.e. $N''$ is a constant
1).  Second, \ti{both} $\Sup{Out}{rlx}(N',N'')$ \ti{and}
$\mi{Out}(N',N'')$ contain assignment $(z'=1,z''=0)$ and hence $N'$
and $N''$ are \ti{inequivalent}.  Separating these two cases comes
down to checking if $N'$ and $N''$ can evaluate to 1 and 0
respectively. If the latter is true, $N'$ and $N''$ are inequivalent.

Set \Sub{D}{out} is built in EC by LoR by computing a sequence of
so-called \tb{boundary formulas} $H_0,\dots,H_k$.  Formula $H_i$
depends only on the variables of a cut of $N',N''$ (see
Figure~\ref{fig:bug_hunt}) and \ti{excludes} the assignments of
$\Sub{D}{Cut}_i = \SUp{Cut}{i}{rlx}(N',N'') \setminus
\mi{Cut}_i(N',N'')$. (That is every assignment of $\Sub{D}{Cut}_i$
\ti{falsifies} $H_i$).  Here $\SUp{Cut}{i}{rlx}(N',N'')$ and
$\mi{Cut}_i(N',N'')$ are the sets of all cut assignments produced when
inputs of $N'$ and $N''$ are unconstrained or constrained by
$\mi{EQ}(X',X'')$ respectively.  Formula $H_0$ is specified in terms
of cut $X' \cup X''$ and is equal to $EQ(X',X'')$. Formula $H_k$ is
computed in terms of cut \s{z',z''} and so specifies the required set
\Sub{D}{out} (as a set of assignments falsifying $H_k$).  Boundary
formulas are computed by a technique called \ti{partial quantifier
  elimination} (PQE) introduced in~\cite{hvc-14}. In PQE, only a part
of the formula is taken out of the scope of quantifiers.  So PQE can
be dramatically more efficient than complete quantifier elimination.

%
%
\subsection{The appeal of EC by LoR}
\label{ssec:appeal}
The appeal of EC by LoR is twofold. First, EC by LoR facilitates
generation of very robust proofs by induction via construction of
boundary formulas.  By contrast, the current approaches (see
e.g.~\cite{kuehlmann97,kuehlmann02,berkeley}) employ fragile induction
proofs e.g. those that require existence of functionally equivalent
internal points.  The size of boundary formulas depends on the
similarity of $N'$ and $N''$ rather than their individual complexity.
This suggests that proofs of equivalence in EC by LoR can be generated
efficiently.

\setlength{\intextsep}{4pt}
\begin{wrapfigure}{l}{1.7in}
 \begin{center}
    \includegraphics[width=1.6in]{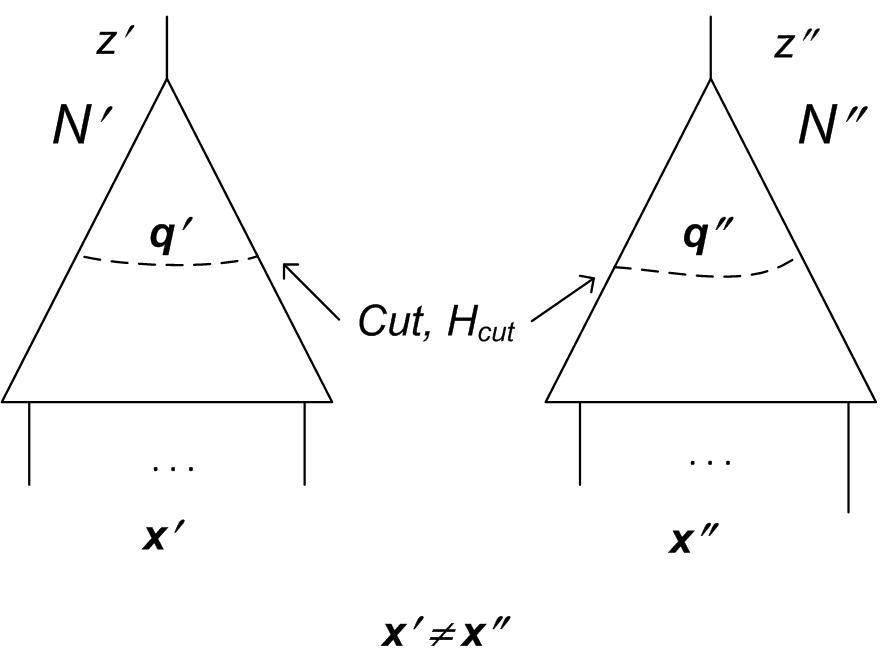}
  \end{center}
\vspace{-10pt}
\caption{Building boundary formula \Sub{H}{cut}}
\label{fig:bug_hunt}
\end{wrapfigure}

Second, the machinery of boundary formulas facilitates proving
inequivalence.  Let $\Sub{F}{N'}(X',Y',z')$ and
$\Sub{F}{N''}(X'',Y'',z'')$ be formulas specifying $N'$ and $N''$
respectively.  (We will say that a Boolean formula $F_N$ specifies
circuit $N$ if every assignment satisfying $F_N$ is a consistent
assignment to variables of $N$ and vice versa. We will assume that all
formulas mentioned in this paper are Boolean formulas in Conjunctive
Normal Form (CNF) unless otherwise stated.)

Circuits $N'$ and $N''$ are inequivalent iff formula $\mi{EQ}(X',X'')
\wedge F_{N'} \wedge F_{N''} \wedge (z' \not\equiv z'')$ is
satisfiable. Denote this formula as \al.  As we show in this paper,
\Al is equisatisfiable with formula \Be equal to $\Sub{H}{cut} \wedge
F_{N'} \wedge F_{N''} \wedge (z' \not\equiv z'')$. Here
\Sub{H}{cut}\ is a boundary formula computed with respect to a cut
(see Fig.~\ref{fig:bug_hunt}.) In general, formula \Be is easier to
satisfy than \Al for the following reason. Let \pnt{p} be an
assignment satisfying formula \Be. Let \pnt{x'} and \pnt{x''} be the
assignments to variables of $X'$ and $X''$ respectively specified by
\pnt{p}.  Since variables of $X'$ and $X''$ are independent of each
other in formula \be, in general, $\pnt{x'}\neq \pnt{x''}$ and so
\pnt{p} does not satisfy \al. Hence, neither \pnt{x'} nor \pnt{x''}
are a counterexample. They are just inputs producing cut assignments
\pnt{q'} and \pnt{q''} (see Fig.~\ref{fig:bug_hunt}) such that a)
$\Sub{H}{cut}(\pnt{q'},\pnt{q''})=1$ and b) $N'$ and $N''$ produce
different outputs under cut assignment (\pnt{q'},\pnt{q''}). To turn
\pnt{p} into an assignment satisfying \Al one has to do \ti{extra
  work}. Namely, one has to find assignments \pnt{x'} and \pnt{x''} to
$X'$ and $X''$ that are \ti{equal to each other} and under which $N'$
and $N''$ produce cut assignments \pnt{q'} and \pnt{q''} above.  Then
\pnt{x'} and \pnt{x''} specify a counterexample. So the
equisatisfiability of \Al and \Be allows one to prove $N'$ and $N''$
inequivalent (by showing that \Be is satisfiable) without providing a
counterexample.

%
%
\subsection{Contributions and structure of the paper}
Our contributions are as follows.  First, we present a new method of
EC based on LoR meant for a very general class of structurally similar
circuits. This method is formulated in terms of a new technique called
PQE that is a ``light'' version of quantifier elimination.  Showing
the potential of PQE for building new verification algorithms is our
second contribution.  Third, we relate EC by LoR to existing methods
based on finding equivalent internal points. In particular, we show
that a set of clauses relating points of an equivalence cut of $N'$
and $N''$ is a boundary formula. So boundary formulas can be viewed as
a machinery for generalization of the notion of an equivalence
cut. Fourth, we give experimental evidence in support of EC by LoR. In
particular, we employ an existing PQE algorithm whose performance can
be drastically improved for solving non-trivial EC problems. Fifth,
we show that interpolation is a special case of LoR (and interpolants
are a special case of boundary formulas.)  In particular, we
demonstrate that by using LoR one can interpolate a ``broken''
implication. This extension of interpolation can be used for
generation of short versions of counterexamples.

The structure of this paper is as follows. Section~\ref{sec:appeal}
discusses the challenge of proving EC by induction.  In
Section~\ref{sec:rel_lor_pqe}, we show the correctness of EC by LoR
and relate the latter to partial quantifier elimination.  Boundary
formulas are discussed in Section~\ref{sec:bnd_form}.
Section~\ref{sec:ec_by_lor} presents an algorithm of EC by LoR.
Section~\ref{sec:bnd_form_approx} describes how one can apply EC by
LoR if the power of a PQE solver is not sufficient to compute boundary
formulas precisely. Section~\ref{sec:experiments} provides
experimental evidence in favor of our approach.  In
Section~\ref{sec:background}, some background is given.  We relate
interpolation and LoR in Section~\ref{sec:lor_interp} and make
conclusions in Section~\ref{sec:conclusions}. The appendix of the
paper contains six  sections with additional information.

\section{Proving Equivalence By Induction}
\label{sec:appeal}
Intuitively, for structurally similar circuits $N'$ and $N''$, there
should exist a short proof of equivalence shown in
Fig.~\ref{fig:ind_proof}.  In this proof, for every set $\mi{Cut}_i$
forming a cut, only a small set $H_i$ of short clauses relating
variables of $\mi{Cut}_i$ is generated. (A \tb{clause} is a
disjunction of literals. We will use the notions of a CNF formula $C_1
\wedge .. \wedge C_p$ and the set of clauses \s{C_1,\dots,C_p}
interchangeably). The relations of $i$-th cut specified by $H_i$ are
derived using formulas $H_j$ built earlier i.e. $j < i$. This goes on
until clauses specifying $z' \equiv z''$ are derived. We will refer to
the proof shown in Fig.~\ref{fig:ind_proof} as a \tb{proof by
  induction} (slightly abusing the term ``induction'').  A good
scalability of the current EC tools is based on their ability to
derive proofs by induction.  However, they can find such proofs only
when cut variables have a very tight relation (most commonly, an
equivalence relation). This means that these tools can handle only a
very narrow subclass of structurally similar circuits.

Proving EC by induction is a challenging task because one has to
address the following \tb{cut termination problem}.  When does one
stop generating a set of clauses $H_i$ in terms of variables of
$\mi{Cut}_i$ and switch to building formula $H_{i+1}$ relating
variables of $\mi{Cut}_{i+1}$?  Let $M_i$ denote the subcircuit
consisting of the gates of $N'$ and $N''$ located below $i$-th cut
(like subcircuit $M$ of Fig.~\ref{fig:bnd_form}). A straightforward
way to build an inductive proof is to make formula $H_i$ specify the
\ti{range} of $M_i$ i.e. the set of all output assignments that can be
produced by $M_i$.  (We will also refer to the range of circuit $M_i$
as a \tb{cut image} because it specifies all assignments that can
appear on $i$-th cut.)  Then formula $H_{i+1}$ can be derived from
formula $H_i$ and the clauses specifying the gates located between
$\mi{Cut}_i$ and $\mi{Cut}_{i+1}$.  A flaw of this approach is that a
formula specifying the image of $i$-th cut can get prohibitively
large.

\setlength{\intextsep}{4pt}
\begin{wrapfigure}{l}{1.7in}
 \begin{center}
    \includegraphics[width=1.7in]{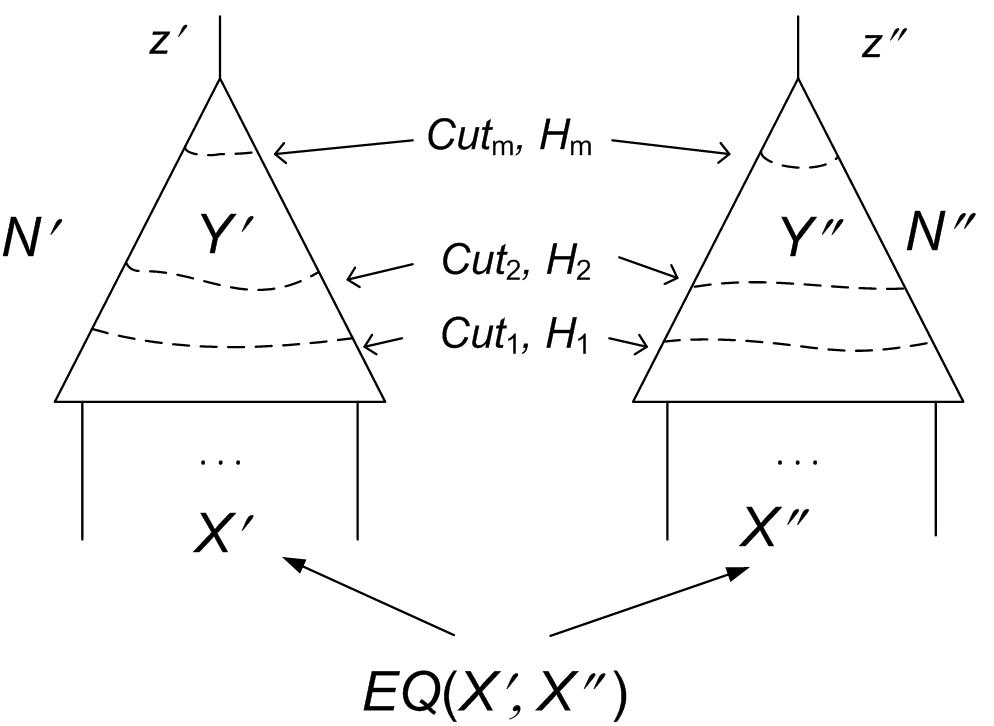}
  \end{center}
\vspace{-5pt}
\caption{An inductive proof of equivalence}
\label{fig:ind_proof}
\end{wrapfigure}

A solution offered in EC by LoR is to use the boundary formulas
introduced in Subsection~\ref{ssec:ec_by_lor} as formulas $H_i$. This
solution has at least three nice qualities. First, boundary formulas
have simple semantics. ($H_i$ excludes the assignments of
\mbox{$i$-th} cut that can be produced when inputs of $N'$ and $N''$
are independent of each other but cannot be produced when inputs are
constrained by $\mi{EQ}(X',X'')$.)  Second, the size of a boundary
formula depends on the structural similarity of circuits $N'$ and
$N''$ rather than their individual complexity. In other words, a
boundary formula computed for a cut is drastically simpler than a
formula specifying the image of this cut in $N'$ and $N''$. Third,
formula $H_i$ can be inductively derived from $H_{i-1}$, which gives
an elegant solution to the cut termination problem. The construction
of formula $H_i$ ends (and that of $H_{i+1}$ begins) when adding $H_i$
to some quantified formula containing $H_{i-1}$ makes the latter
\ti{redundant}.

\section{Equivalence Checking By LoR And PQE}
\label{sec:rel_lor_pqe}
In this section, we prove the correctness of Equivalence Checking (EC)
by Lo-gic Relaxation (LoR) and relate the latter to Partial Quantifier
Elimination (PQE).  Subsection~\ref{subsec:pqe_def} introduces PQE.
In Subsection~\ref{subsec:ec_lor_pqe}, we discuss proving
equivalence/inequivalence in EC by LoR. Besides, we relate EC by LoR to
PQE.

%
%
\subsection{Complete and partial quantifier elimination}
\label{subsec:pqe_def}
In this paper, by a quantified formula we mean one with
\ti{existential} quantifiers.  Given a quantified formula
\prob{W}{A(V,W)}, the problem of \ti{quantifier elimination} is to
find a quantifier-free formula $A^*(V)$ such that $A^* \equiv
\prob{W}{A}$.  Given a quantified formula \prob{W}{A(V,W) \wedge
  B(V,W)}, the problem of \tb{Partial Quantifier Elimination}
(\tb{PQE}) is to find a quantifier-free formula $A^*(V)$ such that
\pqe{W}{A}{A^*}{B}.  Note that formula $B$ remains quantified (hence
the name \ti{partial} quantifier elimination). We will say that
formula $A^*$ is obtained by \tb{taking} \pnt{A} \tb{out of the scope
  of quantifiers} in \prob{W}{A \wedge B}. Importantly, there is a
strong relation between PQE and the notion of \ti{redundancy} of a
subformula in a quantified formula. In particular, solving the PQE
problem above comes down to finding $A^*(V)$ implied by $A \wedge B$
that makes $A$ redundant in $A^* \wedge \prob{W}{A \wedge B}$. Indeed,
in this case, $\prob{W}{A \wedge B} \equiv A^* \wedge \prob{W}{A
  \wedge B} \equiv A^* \wedge \prob{W}{B}$.

Importantly, redundancy with respect to a quantified formula is
\ti{much more powerful} than that with respect to a quantifier-free
one.  For instance, if formula $F(V)$ is satisfiable, every clause of
$F$ is redundant in formula \prob{V}{F}.  On the other hand, a clause
$C$ is redundant in a \ti{quantifier-free} formula $F$ \ti{only} if
$C$ is implied by $F \setminus \s{C}$.

Let $G(V)$ be a formula implied by $B$. Then \pqe{W}{A}{A^* \wedge
  G}{B} entails \pqe{W}{A}{A^*}{B}. In other words, clauses implied by
the formula that remains quantified are \ti{noise} and can be removed
from a solution to the PQE problem. So when building $A^*$ by
resolution it is sufficient to use only the resolvents that are
descendants of clauses of $A$. For that reason, in the case formula
$A$ is much smaller than $B$, PQE can be dramatically faster than
complete quantifier elimination.  Another way to contrast complete
quantifier elimination with PQE is as follows.  The former deals with
a single formula and so, in a sense, has to cope with its
\ti{absolute} complexity. By contrast, PQE computes formula $A^*$ that
specifies the ``difference'' between formulas $B$ and $A \wedge B$.
So the efficiency of PQE depends on their \ti{relative} complexity.
This is important because no matter how high the individual complexity
of $B$ and $A \wedge B$ is, their relative complexity can be quite
manageable.  In Section~\ref{app:pqe} of the appendix we briefly
describe an algorithm for PQE and recall some relevant
results~\cite{fmcad12,fmcad13,hvc-14}.

%
%
\subsection{Proving equivalence/inequivalence by LoR}
\label{subsec:ec_lor_pqe}
Proposition~\ref{prop:ec_by_lor} below shows how one
proves\footnote{The proofs of propositions are given in
  Section~\ref{app:proofs} of the appendix.} equivalence/inequivalence
of circuits by LoR.  Let formula $G$ denote $\mi{EQ} \wedge
\Sub{F}{N'} \wedge \Sub{F}{N''}$ and formula \Sup{G}{rlx} denote
$\Sub{F}{N'} \wedge \Sub{F}{N''}$. Recall from
Subsection~\ref{ssec:appeal} that $F_{N'}(X',Y',z')$ and
$F_{N''}(X'',Y'',z'')$ specify circuits $N'$ and $N''$
respectively. Formula $\mi{EQ}(\pnt{x'},\pnt{x''})$ evaluates to 1 iff
\pnt{x'}=\pnt{x''} where \pnt{x'} and \pnt{x''} are assignments to
variables of $X'$ and $X''$ respectively.
%
%
\begin{proposition}
\label{prop:ec_by_lor}
Let $H(z',z'')$ be a formula such that
\pqe{W}{\mi{EQ}}{H}{\Sup{G}{rlx}} where $W = X' \cup X'' \cup Y' \cup
Y''$. Then formula $G \wedge (z' \not\equiv z'')$ is equisatisfiable
with $H \wedge \Sup{G}{rlx} \wedge (z' \not\equiv z'')$.
\end{proposition}

Note that finding formula $H(z',z'')$ of
Proposition~\ref{prop:ec_by_lor} reduces to taking formula $\mi{EQ}$
out of the scope of quantifiers i.e. to solving the PQE problem.
Proposition~\ref{prop:ec_by_lor} implies that proving
\ti{inequivalence} of $N'$ and $N''$ comes down to showing that
formula \Sup{G}{rlx} is satisfiable under assignment $(z'=b', z''=
b'')$ (where $b',b'' \in \s{0,1}$) such that $b' \neq b''$ and
$H(b',b'')=1$. Recall that the input variables of $N'$ and $N''$ are
independent of each other in formula \Sup{G}{rlx}. Hence the only
situation where \Sup{G}{rlx} is unsatisfiable under $(z'=b', z''=
b'')$ is when $N'$ is constant $\overline{b'}$ and/or $N''$ is
constant $\overline{b''}$.  So the corollary below holds.
%
%
\begin{corollary}
If neither $N'$ nor $N''$ are constants, they are equivalent iff
$H(1,0)=H(0,1)=0$. 
\end{corollary}

Reducing EC to an instance of PQE also provides valuable information
when proving \ti{equivalence} of $N'$ and $N''$.  Formula \Sup{G}{rlx}
remains quantified in \pqe{W}{\mi{EQ}}{H}{\Sup{G}{rlx}}. This means
that to obtain formula $H$, it suffices to generate only resolvents
that are descendants of clauses of $\mi{EQ}$. The clauses obtained by
resolving solely clauses of \Sup{G}{rlx} are just ``noise'' (see
Subsection~\ref{subsec:pqe_def}). This observation is the basis for
our algorithm of generating EC proofs by induction.

\section{Boundary Formulas}
\label{sec:bnd_form}
In this section, we discuss boundary formulas, a key notion of EC by
LoR.  Subsection~\ref{subsec:bf_semant} explains the semantics of
boundary formulas.  Subsection~\ref{subsec:bf_size} discusses the size
of boundary formulas. In Subsection~\ref{ssec:ind_comp_bf}, we
describe how boundary formulas are built.

%
%
\subsection{Definition and some properties of boundary formulas}
\label{subsec:bf_semant}
Let $M$ be the subcircuit consisting of the gates of $N',N''$ located
before a cut as shown in Fig.~\ref{fig:bnd_form}.  As usual, $G$
denotes $\mi{EQ}(X',X'') \wedge F_{N'} \wedge F_{N''}$ and
\Sup{G}{rlx} does $F_{N'} \wedge F_{N''}$.
%
%
\begin{definition}
\label{def:bnd_form}
Let formula \Sub{H}{cut} depend only on variables of a cut.  Let
\pnt{q} be an assignment to the variables of this cut. Formula
\Sub{H}{cut} is called \tb{boundary} if\footnote{ Since formula $(z' \not\equiv z'')$
constraining the outputs of $N'$ and $N''$ is not a part of formulas
\Sup{G}{rlx} and $G$, a boundary formula of
Definition~\ref{def:bnd_form} is not ``property driven''. This can be
fixed by making a boundary formula specify the difference between
$\Sup{G}{rlx}\wedge (z' \not\equiv z'')$ and $G \wedge (z' \not\equiv
z'')$ rather than between \Sup{G}{rlx} and $G$. In this paper, we
explore boundary formulas of Definition~\ref{def:bnd_form}.  The only
exception is Section~\ref{sec:lor_interp} where, to compare LoR and
interpolation, we use ``property-driven'' boundary formulas.

}
\begin{enumerate}[a)]
\vspace{-2pt}
\item \Impl{G}{\Sub{H}{cut}} holds and
\item for every \pnt{q} that can be extended to satisfy \Sup{G}{rlx} but
  cannot be extended to satisfy $G$, the value of
  \Sub{H}{cut}(\pnt{q}) is 0.
\end{enumerate}
\end{definition}

Note that Definition~\ref{def:bnd_form} does not specify the value of
\Sub{H}{cut}(\pnt{q}) if \pnt{q} \ti{cannot} be extended to satisfy
\Sup{G}{rlx} (and hence $G$). As we mentioned in the introduction,
formula $\mi{EQ}(X',X'')$ and formula $H(z',z'')$ of
Proposition~\ref{prop:ec_by_lor} are actually boundary formulas with
respect to cuts $X' \cup X''$ and \s{z',z''} respectively. We will
refer to $H(z',z'')$ as an \tb{output boundary formula}.
Proposition~\ref{prop:bf_by_pqe} below reduces building \Sub{H}{cut}
to PQE.

\setlength{\intextsep}{4pt}
\begin{wrapfigure}{l}{1.7in}
 \begin{center}
    \includegraphics[width=1.6in]{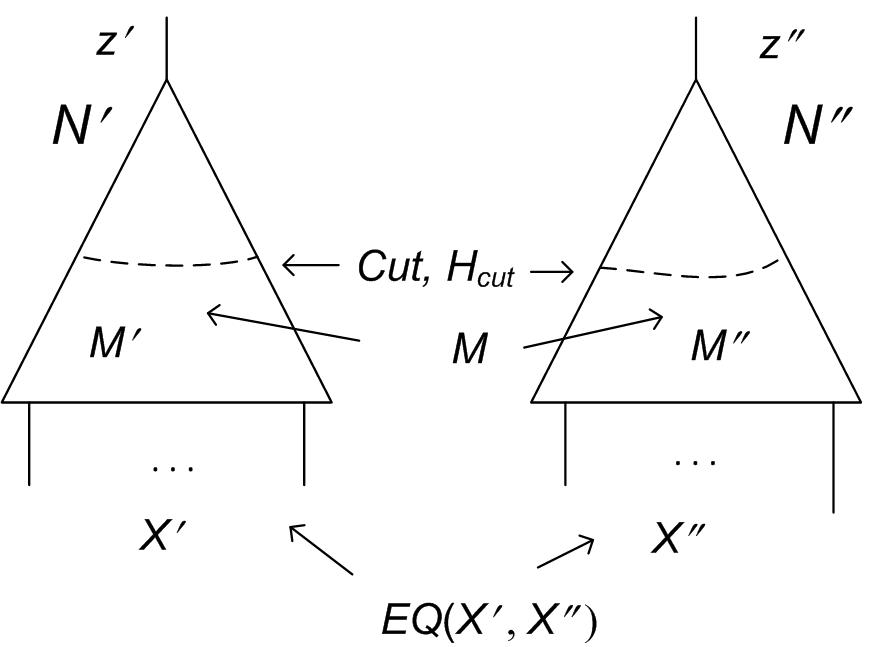}
  \end{center}
\caption{Building boundary formula \Sub{H}{cut}}
\label{fig:bnd_form}
\end{wrapfigure}

%

%
%
\begin{proposition}
\label{prop:bf_by_pqe}
Let \Sub{H}{cut} be a formula depending only on variables of a cut.
Let \Sub{H}{cut} satisfy \pqe{W}{\mi{EQ}}{\Sub{H}{cut}}{F_M}.  Here
$W$ is the set of variables of $F_M$ minus those of the cut.  Then
\Sub{H}{cut} is a boundary formula.
\end{proposition}

 Proposition~\ref{prop:bf_prop} below extends
 Proposition~\ref{prop:ec_by_lor} to an arbitrary boundary formula.

%
%
\begin{proposition}
\label{prop:bf_prop}
Let \Sub{H}{cut} be a boundary formula with respect to a cut.  Then $G
\wedge (z' \not\equiv z'')$ is equisatisfiable with $\Sub{H}{cut}
\wedge \Sup{G}{rlx} \wedge (z' \not\equiv z'')$.
\end{proposition}
%

%
The proposition below estimates the size of boundary formulas built
for $N'$ and $N''$ that satisfy the notion of structural similarity
introduced in Subsection~\ref{ssec:ssim_circs}.
\begin{proposition}
\label{prop:ssim_circs}
Let $\mi{Cut'},\mi{Cut}''$ specify the outputs of circuits $M'$ and
$M''$ of Fig.~\ref{fig:bnd_form} respectively. Assume that for every
variable $v'$ of $\mi{Cut'}$ there is a set $S(v') =
\s{v''_{i_1},\dots,v''_{i_k}}$ of variables of $\mi{Cut''}$ that have
the following property.  Knowing the values of variables of $S(v')$
produced in $N''$ under input \pnt{x} one can determine the value of
$v'$ of $N'$ under the same input \pnt{x}. We assume here that $S(v')$
has this property for every possible input \pnt{x}. Let
$\mi{Max}(S(v'))$ be the size of the largest $S(v')$ over variables of
$\mi{Cut'}$. Then there is a boundary formula \Sub{H}{cut} where every
clause has at most $\mi{Max}(S(v'))+1$ literals.
\end{proposition}

Proposition~\ref{prop:ssim_circs} demonstrates the existence of small
boundary formulas for structurally similar circuits $N'$,$N''$.
Importantly, the size of these boundary formulas depend on
\ti{similarity} of $N'$ and $N''$ rather than their \ti{individual
  complexity}.

%
%
\begin{corollary}
\label{cor:eq_cut}
Let circuits $M'$ and $M''$ of Fig.~\ref{fig:bnd_form} be functionally
equivalent. Then for every variable $v' \in \mi{Cut'}$ there is a set
$S(v')=\s{v''}$ where $v''$ is the variable of $\mi{Cut''}$ that is
functionally equivalent to $v'$. In this case, formula
$\mi{EQ}(\mi{Cut}',\mi{Cut}'')$ stating equivalence of corresponding
output variables of $M'$ and $M''$ is a boundary formula for the cut
in question. This formula can be represented by $2\!*\!p$ two-literal
clauses where $p = |Cut'|=|Cut''|$.
\end{corollary}

Proposition~\ref{prop:ssim_circs} and the corollary above show that
the machinery of boundary formulas allows one to extend the notion of
an equivalence cut to the case where structurally similar circuits
have no functionally equivalent internal variables.

%
%
\subsection{Size of boundary formulas in general case}
\label{subsec:bf_size}
Proposition~\ref{prop:ssim_circs} above shows the existence of small
boundary formulas for a particular notion of structural similarity. In
this subsection, we make two observations that are applicable to a
more general class of structurally similar circuits than the one
outlined in Subsection~\ref{ssec:ssim_circs}.

The first observation is as follows.  Let \pnt{q} be an assignment to
the cut of Fig.~\ref{fig:bnd_form}. Assignment \pnt{q} can be
represented as (\pnt{q'},\pnt{q''}) where \pnt{q'} and \pnt{q''} are
assignments to output variables of $M'$ and $M''$
respectively. Definition~\ref{def:bnd_form} does not constrain the
value of \Sub{H}{cut}(\pnt{q}) if \pnt{q} cannot be extended to
satisfy $F_{N'} \wedge F_{N''}$. So, if, for instance, output \pnt{q'}
cannot be produced by $M'$ for any input, the value of
\Sub{H}{cut}(\pnt{q}) can be arbitrary.  This means that \Sub{H}{cut}
does not have to tell apart cut assignments that \ti{can} be produced
by $M'$ and $M''$ from those that \ti{cannot}. In other words,
\Sub{H}{cut} does not depend on the \ti{individual complexity} of $M'$
and $M''$. Formula \Sub{H}{cut} has only to differentiate cut
assignments that can be produced \ti{solely} when $\pnt{x'} \neq
\pnt{x''}$ from those that can be produced when
\pnt{x'}=\pnt{x''}. Here \pnt{x'} and \pnt{x''} are assignments to
$X'$ and $X''$ respectively.

The second observation is as follows. Intuitively, even a very broad
definition of structural similarity of $N'$ and $N''$ implies the
existence of many short clauses relating cut variables that can be
derived from $\mi{EQ}(X',X'') \wedge F_M$. These clauses can be
effectively used to eliminate the output assignments of $M$ that can
be produced only by inputs (\pnt{x'},\pnt{x''}) where $\pnt{x'} \neq
\pnt{x''}$. Proposition~\ref{prop:ssim_circs} above substantiates this
intuition in case the similarity of $N'$ and $N''$ is defined as in
Subsection~\ref{ssec:ssim_circs}.

\subsection{Computing Boundary Formulas}
\label{ssec:ind_comp_bf}
The key part of EC by LoR is to compute an output boundary formula
$H(z',z'')$.  In this subsection, we show how to build formula $H$
\ti{inductively} by constructing a sequence of boundary formulas $H_0,
\dots, H_k$ computed with respect to cuts
$\mi{Cut}_0,\dots,\mi{Cut}_k$ of $N'$ and $N''$ (see
Fig.~\ref{fig:ind_proof}). We assume that $\mi{Cut}_0=X' \cup X''$ and
$\mi{Cut}_k = \s{z',z''}$ (i.e. $H = H_k$) and $\mi{Cut}_i \cap
\mi{Cut}_j = \emptyset$ if $i \neq j$.

Boundary formula $H_0$ is set to $\mi{EQ}(X',X'')$ whereas formula
$H_i$, $i > 0$ is computed from $H_{i-1}$ as follows. Let $M_i$ be the
circuit consisting of the gates located between the inputs of $N'$ and
$N''$ and cut $\mi{Cut_i}$ (as circuit $M$ of
Fig.~\ref{fig:bnd_form}). Let $F_{M_i}$ be the subformula of
\Sup{G}{rlx} specifying $M_i$.  Let $W_i$ consist of all the variables
of $F_{M_i}$ minus those of $\mi{Cut}_i$. Formula $H_i$ is built to
satisfy \pqe{W_i}{H_{i-1}}{H_i}{F_{M_i}} and so make the previous
boundary formula $H_{i-1}$ \ti{redundant} in $H_i \wedge
\prob{W_i}{H_{i-1} \wedge F_{M_i}}$.  The fact that $H_1,\dots,H_k$
are indeed boundary formulas follows from Proposition~\ref{prop:part}.

%
%
\begin{proposition}
\label{prop:part}
Let $W_i$ where $i > 0$ be the set of variables of $F_{M_i}$ minus
those of $\mi{Cut}_i$.  Let $H_{i-1}$ where $i > 1$ be a boundary
formula such that $\prob{W_{i-1}}{H_0 \wedge F_{M_{i-1}}} \equiv
H_{i-1} \wedge \exists{W_{i-1}} [F_{M_{i-1}}]$.  Let
\pqe{W_i}{H_{i-1}}{H_i}{F_{M_i}} hold. Then \PRob{W_i}{H_0 \wedge
  F_{M_i}} $\equiv$ $H_i \wedge \prob{W_i}{F_{M_i}}$ holds. (So $H_i$
is a boundary formula due to Proposition~\ref{prop:bf_by_pqe}.)
\end{proposition}

\section{Algorithm of EC by LoR}
\label{sec:ec_by_lor}
In this section, we introduce an algorithm called \EC that checks for
equivalence two single-output circuits $N'$ and $N''$. The pseudo-code
of \EC is given in Figure~\ref{fig:ec_by_lor}.  \EC builds a sequence
of boundary formulas $H_0,\dots,H_k$ as described in
Subsection~\ref{ssec:ind_comp_bf}.  Here $H_0$ equals
$\mi{EQ}(X',X'')$ and $H_k(z',z'')$ is an output boundary
formula. Then, according to Proposition~\ref{prop:ec_by_lor}, \EC
checks the satisfiability of formula $H_k \wedge \Sup{G}{rlx} \wedge
(z' \not\equiv z'')$ where $\Sup{G}{rlx} = F_{N'} \wedge F_{N''}$.

\setlength{\intextsep}{2pt}
\begin{wrapfigure}{L}{2.1in}
\small
\begin{tabbing}
aaa\=bb\=cc\= dd\= \kill
$\mi{EC\!\_LoR}(N',N'')$\{\\
\tb{\scriptsize{1}}\> $(N',N'') := \mi{Bufferize}(N',N'')$; \\
\tb{\scriptsize{2}}\> $\mi{Cut}_0 = X' \cup X''$; \\
\tb{\scriptsize{3}}\> $\mi{Cut}_k := \s{z',z''}$; \\
\tb{\scriptsize{4}}\> $\mi{Cut}_1,..,\mi{Cut}_{k-1}\!:=\!\mi{BldCuts}(N',N'')$; \\
\tb{\scriptsize{5}}\> $H_0 := \mi{EQ}(X',X'')$; \\
$----------------$ \\
\tb{\scriptsize{6}}\> $\mi{for} (i:=1; i \leq k; i\mi{+\!+})$ \{  \\
\tb{\scriptsize{7}}\Tt $H_i = 1$;  \\
\tb{\scriptsize{8}}\Tt $F_{M_i} := \mi{SubForm}(\Sup{G}{rlx},\mi{Cut}_i)$; \\
\tb{\scriptsize{9}}\Tt $W_i := \mi{Vars}(F_{M_i}) \setminus \mi{Vars}(\mi{Cut}_i)$; \\
\tb{\scriptsize{10}}\Tt while ($\mi{true}$) \{ \\
\tb{\scriptsize{11}}\ttt  $C\!:=\!Redund(H_i\!\wedge\!\prob{W_i}{H_{i-1}\!\wedge\!F_{M_i}})$;\\
\tb{\scriptsize{12}}\ttt  if ($C = \mi{nil}$) break;\\
\tb{\scriptsize{13}}\ttt  $H_i := H_i \wedge C$;\}\} \\
$----------------$ \\
\tb{\scriptsize{14}}\> if ($H_k(0,1) = 1$) \\
\tb{\scriptsize{15}}\Tt if $(Sat(\Sup{G}{rlx} \wedge \overline{z'} \wedge z''))$ return(\ti{No}); \\
\tb{\scriptsize{16}}\> if ($H_k(1,0) = 1$) \\
\tb{\scriptsize{17}}\Tt if $(Sat(\Sup{G}{rlx} \wedge z' \wedge \overline{z''}))$ return(\ti{No}); \\
\tb{\scriptsize{18}}\> return(\ti{Yes}); \}\\
\end{tabbing} 
\vspace{-15pt}
\caption{EC by LoR}
\vspace{7pt}
\label{fig:ec_by_lor}
\end{wrapfigure}

\EC consists of three parts separated by the dotted lines in
Figure~\ref{fig:ec_by_lor}.  \EC starts the first part (lines 1-5) by
calling procedure $\mi{Bufferize}$ that eliminates non-local
connections of $N'$ and $N''$ i.e. those that span more than two
consecutive topological levels. (The \ti{topological level} of a gate
$g$ of a circuit $K$ is the longest path from an input of $K$ to $g$
measured in the number of gates on this path.) The presence of
non-local connections makes it hard to find cuts that do not
overlap. To avoid this problem, procedure $\mi{Bufferize}$ replaces
every non-local connection spanning $d$ topological levels ($d > 2$)
with a chain of $d-2$ buffers. (A more detailed discussion of this
topic is given in Section~\ref{app:cuts_levels} of the appendix.) Then
\EC sets the initial and final cuts to $X' \cup X''$ and \s{z',z''}
respectively and computes the intermediate cuts.

Boundary formulas $H_i$, $ 1 \leq i \leq k$ are computed in the second
part (lines 6-13) that consists of a \ti{for} loop. In the third part
(lines 14-18), \EC uses the output boundary formula $H_k(z',z'')$
computed in the second part to decide whether $N',N''$ are equivalent.
If $H_k(b',b'')=1$ where $b' \neq b''$ and \Sup{G}{rlx} is satisfiable
under $z'=b',z''=b''$, then $N',N''$ are inequivalent. Otherwise, they
are equivalent (line 18).

Formula $H_i$ is computed as follows. First, $H_i$ is set to constant
1.  Then, \EC extracts a subformula $F_{M_i}$ of \Sup{G}{rlx} that
specifies the gates of $N'$ and $N''$ located between the inputs and
cut $\mi{Cut}_i$. \EC also computes the set $W_i$ of quantified
variables. The main work is done in a \ti{while} loop (lines
10-13). First, \EC calls procedure $\mi{Redund}$ that is essentially a
PQE-solver.  $\mi{Redund}$ checks if boundary formula $H_{i-1}$ is
redundant in $H_i \wedge \prob{W_i}{H_{i-1} \wedge F_{M_i}}$ (the cut
termination condition).  $\mi{Redund}$ stops as soon as it finds out
that $H_{i-1}$ is not redundant yet. It returns a clause $C$ as the
evidence that at least one clause must be added to $H_i$ to make
$H_{i-1}$ redundant.  If no clause is returned by $\mi{Redund}$, then
$H_i$ is complete and \EC ends the \ti{while} loop and starts a new
iteration of the \ti{for} loop. Otherwise, \EC adds $C$ to $H_i$ and
starts a new iteration of the \ti{while} loop.

\section{Computing Boundary Formulas By Current PQE Solvers}
\label{sec:bnd_form_approx}
To obtain boundary formula $H_i$, one needs to take $H_{i-1}$ out of
the scope of quantifiers in formula \prob{W_i}{H_{i-1} \wedge F_{M_i}}
whose size grows with $i$ due to formula $F_{M_i}$. So a PQE solver
that computes $H_i$ must have good scalability. On the other hand, the
algorithm of~\cite{hvc-14} does not scale well yet. The main problem
here is that learned information is not re-used in contrast to
SAT-solvers effectively re-using learned clauses. Fixing this problem
requires some time because bookkeeping of a PQE algorithm is more
complex than that of a SAT-solver (see the discussion in
Sections~\ref{app:pqe} and~\ref{app:eff_comp_bfs} of the appendix.)
In this section, we describe two methods of adapting EC by LoR to a
PQE-solver that is not efficient enough to compute boundary formulas
\ti{precisely}. (Both methods are illustrated experimentally in
Section~\ref{sec:experiments}.)

One way to reduce the complexity of computing $H_i$ is to use only a
subset of $F_{M_i}$.  For instance, one can discard the clauses of
$F_{M_i}$ specifying the gates located between cuts $\mi{Cut}_0$ and
$\mi{Cut}_p$, $0 < p < i$. In this case, boundary formula $H_i$ is
computed \ti{approximately}.  The downside of this is that
condition~b) of Definition~\ref{def:bnd_form} does not hold anymore
and so EC by LoR becomes \ti{incomplete}. Namely, if $H(b',b'')=1$
where $b' \neq b''$, the fact that \Sup{G}{rlx} is satisfiable under
$z'=b',z''=b''$ does not mean that $N'$ and $N''$ are inequivalent.
Nevertheless, even EC by LoR with approximate computation of boundary
formulas can be a powerful tool for proving $N'$ and $N''$
\ti{equivalent} for the following reason. If $H(1,0)=H(0,1)=0$,
circuits $N'$ and $N''$ are proved equivalent no matter how
intermediate boundary formulas have been built. Importantly, checking
cut termination conditions is a powerful way to structure the proof
even when boundary formulas are computed approximately.  That is,
construction of $H_i$ still ends when it makes $H_{i-1}$ redundant in
formula $H_i \wedge \prob{W_i}{H_{i-1} \wedge F_{M_i}}$.  The only
difference from computing $H_i$ precisely is that formula $F_{M_i}$ is
simplified by discarding some clauses.

Another way to adapt EC by LoR to an insufficiently efficient PQE
solver is as follows. Suppose that the power of a PQE solver is enough
to build \ti{one} intermediate boundary formula $H_i$ \ti{precisely}.
From Proposition~\ref{prop:bf_prop} it follows that formula \Al equal
to $G \wedge (z' \not\equiv z'')$ is equisatisfiable with formula \Be
equal to $\Sub{H}{cut} \wedge \Sup{G}{rlx} \wedge (z' \not\equiv
z'')$. So, to show that $N'$ and $N''$ are inequivalent it is
sufficient to find an assignment satisfying \be. As we argued in
Subsection~\ref{ssec:appeal}, finding such an assignment for \Be is
easier than for \al.

\section{Experiments}
\label{sec:experiments}
In the experiments, we used the PQE solver published in~\cite{hvc-14}
in 2014.  We will refer to this solver as \tb{PQE-14}.  As we
mentioned in Section~\ref{sec:bnd_form_approx}, PQE-14 does not scale
well yet.  So building a full-fledged equivalence checker based on \EC
would mean simultaneously designing a new EC algorithm and a new PQE
solver.  The latter is beyond the scope of our paper (although the
design of an efficient PQE-solver is discussed in
Section~\ref{app:pqe} of the appendix).  On the other hand, PQE-14 is
efficient enough to make a few important points experimentally. In the
experiments described in this section, we employed a new
implementation of PQE-14.

The experiment of Subsection~\ref{subsec:image_versus_bfs} compares
computing cut image with building a boundary formula for this
cut. (Recall that the image of a cut is the set of cut assignments
that can be produced in $N'$ and $N''$ under all possible inputs.)
This experiment also contrasts complete quantifier elimination
employed to compute cut image with PQE.  In
Subsection~\ref{subsec:ec_lor}, we apply \EC to a non-trivial instance
of equivalence checking that is hard for \ti{ABC}, a high-quality
synthesis and verification tool~\cite{abc}.  In
Subsection~\ref{subsec:sat}, we give evidence that boundary formulas
can be used to prove inequivalence more efficiently.

In the experiments, circuits $N'$ and $N''$ to be checked for
equivalence were derived from a circuit computing an output median bit
of a $k$-bit multiplier. We will refer to this circuit as \m{k}. Our
motivation here is as follows. In many cases, the equivalence of
circuits with simple topology and low fanout values can be efficiently
checked by a general-purpose SAT-solver.  This is not true for
circuits involving multipliers. In all experiments, circuits $N'$ and
$N''$ were bufferized to get rid of long connections (see
Section~\ref{sec:ec_by_lor}).

%
%
\subsection{Image computation versus building boundary formulas}
\label{subsec:image_versus_bfs}
%
%
%
\begin{wraptable}{l}{2.5in}
\small
\caption{\ti{Computing cut image and boundary formula. Time
    limit~=~1~hour}}
\vspace{-5pt}
\scriptsize
\begin{center}
\begin{tabular}{|c|c|c|c|c|c|c|} \hline
\#bits  & \#quan. & \#free    & \multicolumn{2}{c|}{cut image} & \multicolumn{2}{c|}{boundary for-}  \\
        & vars    & vars      & \multicolumn{2}{c|}{(QE)} & \multicolumn{2}{c|}{mula (PQE)}  \\ \cline{4-7}
        &         &           & result                &        & result   &    \\ 
        &         &               & size      & (s.)   & size   & (s.) \\ \hline
    8   &   32    &   84          &  3,142    & 4.0    & \tb{242}    & \tb{0.1}  \\ \hline
    9   &   36    &   104         &  4,937    & 13     & \tb{273}    & \tb{0.2}  \\ \hline
    10  &   40    &   126         &  7,243    & 51     & \tb{407}    & \tb{0.3}     \\ \hline
    11  &   44    &   150         &  9,272    & 147    & \tb{532}    & \tb{0.5}     \\ \hline
    12  &   48    &   176         &  14,731   & 497    & \tb{576}    & \tb{0.6}    \\ \hline
    13  &   52    &   206         &  19,261   & 1,299  & \tb{674}    &  \tb{0.9}    \\ \hline
    14  &   56    &   234         &   $*$     &    $*$ & \tb{971}    & \tb{1.5}    \\ \hline
    15  &   60    &   266         &   $*$     &    $*$ & \tb{1,218}  & \tb{2.0}   \\ \hline
    16  &   64    &   300         &   $*$     &    $*$ & \tb{1,411}  & \tb{3.0}   \\ \hline

\end{tabular}                
\end{center}
\label{tbl:qe_pqe}
\end{wraptable}

In this subsection, we compared computation of a boundary formula
\Sub{H}{cut} and that of cut image. We used two identical copies of
circuit \m{k} as circuits $N'$ and $N''$.  As a cut of $N',N''$ we
picked the set of variables of the first topological level (every
variable of this level specifies the output of a gate fed by input
variables of $N'$ or $N''$). Computing cut image comes down to
performing quantifier elimination for formula \prob{W}{\mi{EQ}(X',X'')
  \wedge F_M}. Here $W = X' \cup X''$ and formula $F_M$ specifies the
gates of the first topological level of $N'$ and $N''$. Formula
\Sub{R}{cut} that is logically equivalent to \prob{W}{\mi{EQ} \wedge
  F_M} specifies the cut image.  Computing a boundary formula comes
down to finding \Sub{H}{cut} such that
\pqe{W}{\mi{EQ}}{\Sub{H}{cut}}{F_M} i.e. solving the PQE problem.

The results of the experiment are given in
Table~\ref{tbl:qe_pqe}. Abbreviation QE stands for Quantifier
Elimination. The value of $k$ in \m{k} is shown in the first
column. The next two columns give the number of quantified and free
variables in \prob{W}{\mi{EQ} \wedge F_M}. To compute formula
\Sub{R}{cut} above we used our quantifier elimination program
presented in~\cite{fmcad13}. Formula \Sub{H}{cut} was generated by
PQE-14. To make this comparison fair, formula \Sub{H}{cut} was
computed without applying any EC-specific heuristics (as opposed to
computing boundary formulas in the experiments of
Subsection~\ref{subsec:ec_lor}).  When computing image formula
\Sub{R}{cut} and boundary formula \Sub{H}{cut} we recorded the size of
the result (as the number of clauses) and the run time in seconds. As
Table~\ref{tbl:qe_pqe} shows, formulas \Sub{H}{cut} are much smaller
than \Sub{R}{cut} and take much less time to compute.

%
%
\subsection{An example of equivalence checking by \EC}
\label{subsec:ec_lor}
\setlength{\intextsep}{4pt}
\begin{wrapfigure}{l}{2.3in}
 \begin{center}
    \includegraphics[height=1.7in,width=2.1in]{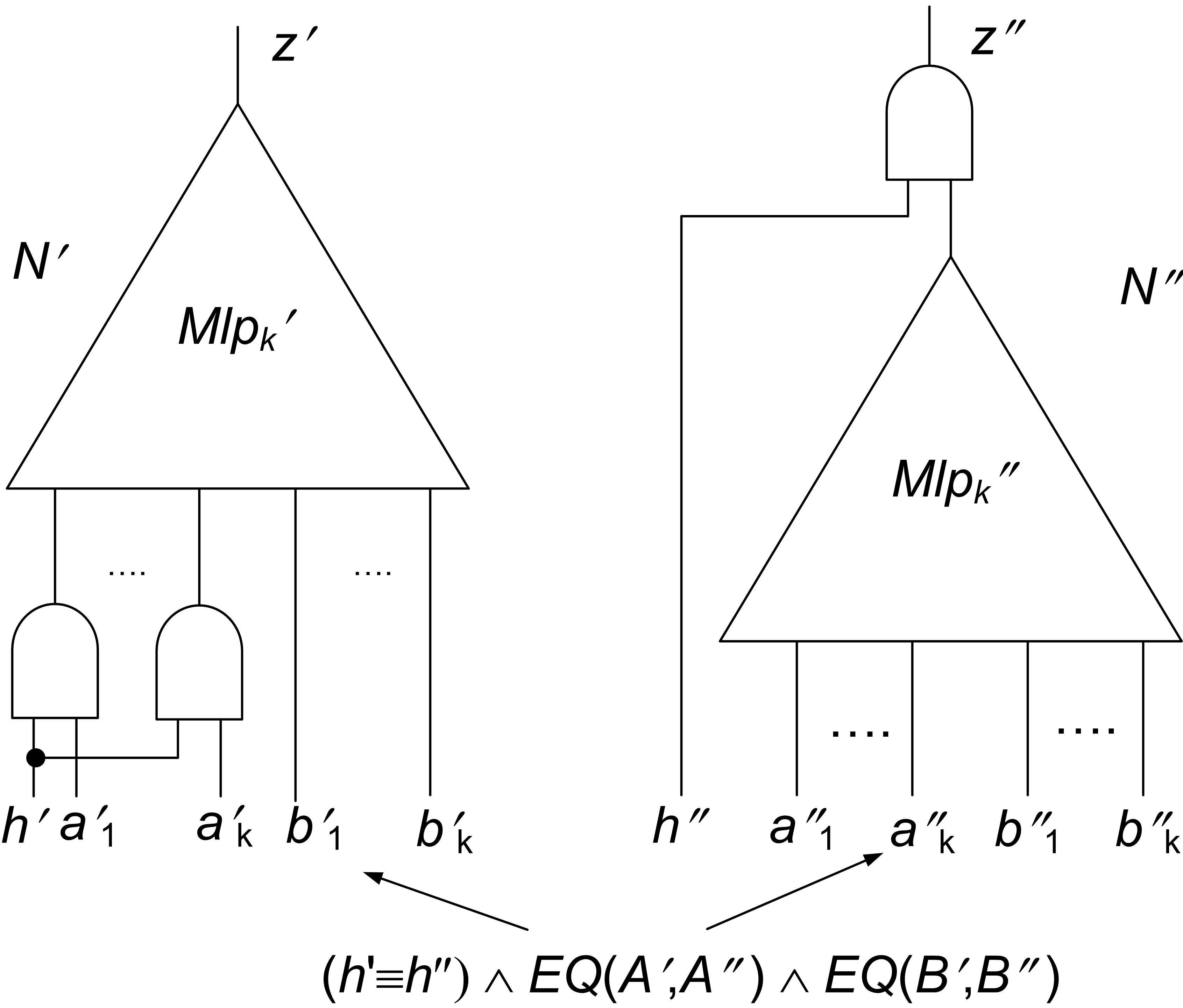}
  \end{center}
\caption{Equivalence checking of $N'$ and $N''$ derived from \m{k}}
\vspace{5pt}
\label{fig:mlps}
\end{wrapfigure}

In this subsection, we run an implementation of \EC introduced in
Section~\ref{sec:ec_by_lor} on circuits $N'$ and $N''$ shown in
Fig.~\ref{fig:mlps}. (The idea of this EC example was suggested by
Vigyan Singhal~\cite{priv_comm}.)  These circuits are derived from
\m{k} by adding one extra input $h$.  Both circuits produce the same
output as \m{k} when $h=1$ and output 0 if $h=0$. So $N'$ and $N''$
are logically equivalent.  Note that the value of every internal
variable of $N'$ depends on $h$ whereas this is not the case for
$N''$.  So $N'$ and $N''$ have no functionally equivalent internal
variables.  On the other hand, $N'$ and $N''$ satisfy the notion of
structural similarity introduced in Subsection~\ref{ssec:ssim_circs}.
Namely, the value of every internal variable $v'$ of $N'$ is specified
by that of $h''$ and some variable $v''$ of $N''$ (So, in this case,
for every internal variable $v'$ of $N'$ there is a set $S(v')$
introduced in Subsection~\ref{ssec:ssim_circs} consisting of only two
variables of $N''$.). In particular, if $v'$ is an internal variable
of $\mi{Mlp}_k'$, then $v''$ is the corresponding variable of
$\mi{Mlp}_k''$.  Indeed, if $h''=1$, then $v'$ takes the same value as
$v''$. If $h''=0$, then $v'$ is a constant (in the implementation of
$\mi{Mlp}_k$ we used in the experiments).  The objective of the
experiment below is to show that \EC can check for equivalence
structurally similar circuits that have no functionally equivalent
internal points.

Cuts $\mi{Cut}_0,\dots,\mi{Cut}_m$ used by \EC \\ were generated
according to topological levels. That is every variable of
$\mi{Cut}_i$ specified the output of a gate of $i$-th topological
level. Since $N'$ and $N''$ were bufferized, $\mi{Cut}_i \cap
\mi{Cut}_j = \emptyset$ if $i \neq j$. The version of \EC we used in
the experiment was slightly different from the one described in
Fig.~\ref{fig:ec_by_lor}. We will refer to this version as \ecm.  The
main change was that boundary formulas were computed in \ECM
\ti{approximately}. That is when checking if formula $H_j$ was
redundant in $H_j \wedge \prob{W_i}{H_{j-1} \wedge F_{M_i}}$ (line 11
of Fig.~\ref{fig:ec_by_lor}) only a subset of clauses of $F_{M_i}$ was
used to make the check simpler. Nevertheless, \ECM was able to compute
an output boundary formula $H(z',z'')$ proving that $N'$ and $N''$
were equivalent. One more difference between \EC and \ECM was as
follows.  \EC runs a cut termination check \ti{every time} formula
$H_i$ is updated (in the \ti{while} loop of Fig.~\ref{fig:ec_by_lor},
lines 10-13). In \ecm, the number of cut termination checks was
reduced. Namely, derivation of clauses of $H_i$ was modified so that
\ECM did not run a cut termination check if some cut variable was not
present in clauses of $H_i$ yet. The intuition here was that in that
case $H_i$ was still under-constrained. \ECM is described in
Section~\ref{app:version} of the appendix in more detail.

\begin{wraptable}{l}{2.5in}
\small
\caption{\ti{EC of $N'$ and $N''$ derived from \m{k}. Time limit = 6
    hours}} \scriptsize
\begin{center}
\begin{tabular}{|c|c|c|c|c|c|} \hline
 \#bits & \#vars    & \#clauses &\#cuts  & \ECM      & \ti{ABC}     \\ 
        &           &           &        &  (s.)     &  (s.)    \\ \hline
   10   &    2,844  &  6,907    &   37   & \tb{4.5}  &  10    \\ \hline
   11   &    3,708  &  8,932    &   41   & \tb{7.1}  &  38    \\ \hline
   12   &    4,726  &  11,297   &   45   & \tb{11}   &  142   \\ \hline
   13   &    5,910  &  14,026   &   49   & \tb{16}   &  757   \\ \hline
   14   &    7,272  &  17,143   &   53   & \tb{25}   &  3,667 \\ \hline
   15   &    8,824  &  20,672   &   57   & \tb{40}   &  11,237 \\ \hline
   16   &   10,578  &  24,637   &   61   & \tb{70}   &  $>$ 21,600     \\ \hline
\end{tabular}                
\end{center}
\label{tbl:ec_by_lor}
\end{wraptable}

In Table~\ref{tbl:ec_by_lor}, we compare \ECM \\ with
\ti{ABC}~\cite{abc}.  The first column gives the value of $k$ of \m{k}
used in $N'$ and $N''$.  The next two columns show the size of
formulas $\mi{EQ}(X',X'') \wedge F_{N'} \wedge F_{N''} \wedge (z'
\not\equiv z'')$ specifying equivalence checking of $N'$ and $N''$ to
which \ECM was applied. (Circuits $N'$ and $N''$ were fed into
\ti{ABC} as circuits in the BLIF format.)  Here $X =
\s{h,a_1,\dots,a_k,b_1,\dots,b_k}$ denotes the set of input variables.
The fourth column shows the number of topological levels in circuits
$N'$ and $N''$ and so the number of cuts used by \ecm.  The last two
columns give the run time of \ECM and \ti{ABC}.

The results of Table~\ref{tbl:ec_by_lor} show that equivalence
checking of $N'$ and $N''$ derived from \m{k} was hard for
\ti{ABC}. On the other hand, \ECM managed to solve all instances in a
reasonable time. Most of the run time of \ECM is taken by PQE-14 when
checking cut termination conditions. So, PQE-14 is also the reason why
the run time of \ECM grows quickly with the size of \m{k}. Using a
more efficient PQE-solver should reduce such a strong dependency of
the performance of \ECM on the value of $k$.

%
%
\subsection{Using boundary formulas for proving inequivalence}
\label{subsec:sat}
In the experiment of this subsection, we checked for equivalence a
correct and a buggy version of \m{16} as circuits $N'$ and $N''$
respectively. Since \ECM described in the previous subsection computes
boundary formulas approximately, one cannot directly apply it to prove
inequivalence of $N'$ and $N''$.  In this experiment we show that the
precise computation of even \ti{one} boundary formula corresponding to
an intermediate cut can be quite useful for proving inequivalence.
Let \Al and \Be denote formulas $\mi{EQ}(X',X'') \wedge F_{N'} \wedge
F_{N''} \wedge (z' \equiv z'')$ and $H_i \wedge F_{N'} \wedge F_{N''}
\wedge (z' \equiv z'')$ respectively. Here $H_i$ is a boundary formula
precisely computed for the cut of $N'$ and $N''$ consisting of the
gates with topological level equal to $i$.  According to
Proposition~\ref{prop:bf_prop}, \Al and \Be are
equisatisfiable. Proving $N'$ and $N''$ inequivalent comes down to
showing that \Be is satisfiable.  Intuitively, checking the
satisfiability of \Be the easier, the larger the value of $i$ and so
the closer the cut to the outputs of $N'$ and $N''$.  In the
experiment below, we show that computing boundary formula $H_i$ makes
proving inequivalence of $N'$ and $N''$ easier even for a cut with a
small value of $i$.

Bugs were introduced into circuit $N''$ \ti{above} the cut (so $N'$
and $N''$ were identical \ti{below} the cut).  Let $M'_i$ and $M''_i$
denote the subcircuits of $N'$ and $N''$ consisting of the gates
located below the cut (like circuits $M'$ and $M''$ in
Fig.~\ref{fig:bnd_form}). Since $M'_i$ and $M''_i$ are identical they
are also functionally equivalent. Then Corollary~\ref{cor:eq_cut}
entails that formula $H_i$ equal to
$\mi{EQ}(\mi{Cut'_i},\mi{Cut''_i})$ is boundary. Here $\mi{Cut'_i}$
and $\mi{Cut''_i}$ specify the output variables of $M'_i$ and $M''_i$
respectively.  Derivation of $\mi{EQ}(\mi{Cut'_i},\mi{Cut''_i})$ for
identical circuits $M'_i$ and $M''_i$ is trivial. However,
\ti{proving} that $H_i$ equal to $\mi{EQ}(\mi{Cut'_i},\mi{Cut''_i})$
is indeed a boundary formula is \ti{non-trivial} even for identical
circuits. (According to Proposition~\ref{prop:bf_by_pqe}, this
requires showing that $\mi{EQ}(X',X'')$ is redundant in $H_i \wedge
\prob{W}{\mi{EQ}(X',X'') \wedge F_{M'_i} \wedge F_{M''_i}}$.)  In
experiments we used cut with $i = 3$ i.e. the gates located below the
cut had topological level less or equal to 3. Proving that
$\mi{EQ}(\mi{Cut'_i},\mi{Cut''_i})$ is a boundary formula takes a
fraction of a second for $i=3$ but requires much more time for $i=4$.

\begin{wraptable}{l}{1.85in}
\small
\caption{\ti{Sat-solving of formulas \Al and \Be by Minisat. Time
    limit = 600 s.}}  \scriptsize
\begin{center}
\begin{tabular}{|c|c|c|c|} \hline
 formula    & \#solv-   & total      &  median  \\ 
  type      & ed        & time (s.)  &  time (s.)    \\ \hline
  \al     &    95       & $>$ 3,490  &  4.2     \\ \hline
  \be    &  \tb{100}   & \tb{1,030} &  \tb{1.0}     \\  \hline
\end{tabular}                
\end{center}
\label{tbl:sat}
\end{wraptable}

We generated 100 buggy versions of \m{16}.  Table~\ref{tbl:sat}
contains results of checking the satisfiability of 100 formulas \Al
and \Be by Minisat~2.0~\cite{minisat,minisat2.0}. Similar results were
observed for the other SAT-solvers we tried. The first column of
Table~\ref{tbl:sat} shows the type of formulas (\Al or \Be). The
second column gives the number of formulas solved in the time limit of
600 s. The third column shows the total run time on all formulas. We
charged 600 s. to every formula \Al that was not solved within the
time limit. The run times of solving formulas \Be include the time
required to build $H_3$.  The fourth column gives the median time. The
results of this experiment show that proving satisfiability of \Be is
noticeably easier than that of \al. Using formula \Be for proving
inequivalence of $N'$ and $N''$ should be much more beneficial if
formula $H_i$ is computed for a cut with a greater value of
$i$. However, this will require a more powerful PQE-solver than
PQE-14.

\section{Some Background}
\label{sec:background}
The EC methods can be roughly classified into two groups. Methods of
the first group do not assume that circuits $N'$ and $N''$ to be
checked for equivalence are structurally similar. Checking if $N'$ and
$N''$ have identical BDDs~\cite{bryant_bdds1} is an example of a
method of this group.  Another method of the first group is to reduce
EC to SAT and run a general-purpose
SAT-solver~\cite{grasp,chaff,minisat,picosat}. A major flaw of these
methods is that they do not scale well with the circuit size.

Methods of the second group try to exploit the structural similarity
of $N',N''$.  This can be done, for instance, by making
transformations that produce isomorphic subcircuits in $N'$ and
$N''$~\cite{beds} or make simplifications of $N'$ and $N''$ that do
not affect their range~\cite{kukula}. The most common approach used by
the methods of this group is to generate an inductive proof by
computing simple relations between internal points of
$N',N''$. Usually, these relations are
equivalences~\cite{kuehlmann97,kuehlmann02,berkeley}. However, in some
approaches the derived relations are implications~\cite{kunz} or
equivalences modulo observability~\cite{brand}. The main flaw of the
methods of the second group is that they are very ``fragile''.  That
is they work only if the equivalence of $N'$ and $N''$ can be proved
by derivation of relations of a very small class.

\section{Logic Relaxation And Interpolation}
\label{sec:lor_interp}
In this section, we compare LoR and interpolation. In
Subsection~\ref{ssec:gen_rlx}, we give a more general formulation of
LoR in terms of arbitrary CNF formulas. In
Subsection~\ref{ssec:spec_case_bf}, we show that interpolation is a
special case of LoR and interpolants are a special case of boundary
formulas. We also explain how one can use LoR to interpolate
a``broken'' implication. This extension of interpolation can be used
for generation of short versions of counterexamples.  Finally, in
Subsection~\ref{ssec:int_bf_ec}, we contrast interpolants with
boundary formulas employed in EC by LoR.

So far we have considered a boundary formula specifying the difference
in assignments satisfying formulas \Sup{G}{rlx} and $G$ equal to
$F_{N'} \wedge F_{N''}$ and $\mi{EQ}(X',X'') \wedge F_{N'} \wedge
F_{N''}$ respectively.  In the footnote of Section~\ref{sec:bnd_form},
we mentioned that one can also consider ``property driven'' boundary
formulas. Such formulas specify the difference in assignments
satisfying $\Sup{G}{rlx} \wedge (z' \not\equiv z'')$ and $G \wedge (z'
\not\equiv z'')$ rather than \Sup{G}{rlx} and $G$. In this section, to
simplify explanation, we use ``property driven'' boundary
formulas. They describe the difference in assignments satisfying a
relaxed formula and an original formula that is supposed to be
unsatisfiable.

%
%
\subsection{Generalizing LoR to arbitrary formulas}
\label{ssec:gen_rlx}
Let $S(X,Z)$ be a formula whose satisfiability one needs to check.
Here $X$ and $Z$ are non-overlapping sets of Boolean variables.  In
the context of formal verification, one can think of $S$ as obtained
by conjoining formulas $G(X,Z)$ and $\overline{\mi{Good}(Z)}$. Here
$G(X,Z)$ specifies the consistent design behaviors, $X$ and $Z$ being
sets of ``internal'' and ``external'' variables. Formula $\mi{Good}$
specifies design behaviors that preserve a required property defined
in terms of external variables.

Let formula $S$ be represented as $\Sup{S}{rlx}(X,Z) \wedge
E(X,Z)$. Formula \Sup{S}{rlx} can be viewed as a relaxation of $S$
that is easier to satisfy. Let $H(Z)$ be a formula obtained by taking
$E$ out of the scope of quantifiers in \prob{X}{E \wedge \Sup{S}{rlx}}
i.e. $\prob{X}{E \wedge \Sup{S}{rlx}} \equiv H \wedge
\prob{X}{\Sup{S}{rlx}}$.  Then $S$ is \ti{equisatisfiable} to $H
\wedge \Sup{S}{rlx}$ (see Proposition~\ref{prop:gen_method} of the
appendix). Checking the satisfiability of $H \wedge \Sup{S}{rlx}$
reduces to testing the satisfiability of \Sup{S}{rlx} under
assignments to $Z$ for which $H$ evaluates to 1. So, if formula $S$ is
``sufficiently'' relaxed in \Sup{S}{rlx} and $Z$ is much smaller than
$X$, solving formula $H \wedge \Sup{S}{rlx}$ can be drastically simply
than $S$.

One can view $H$ as a \ti{boundary formula} specifying the difference
in assignments satisfying $S$ and \Sup{S}{rlx}. In particular, formula
$H$ satisfies the properties of Definition~\ref{def:bnd_form} (see
Proposition~\ref{prop:gen_bf} of the appendix). That is a) $S
\rightarrow H$ and b) $H(\pnt{z})=0$ for every assignment \pnt{z} to
$Z$ that can be extended to satisfy \Sup{S}{rlx} but not $S$.

%
%
\subsection{Interpolation as a special case of LoR}
\label{ssec:spec_case_bf}
Let formula $S$ denote $A(X,Y) \wedge B(Y,Z)$ where $X,Y,Z$ are
non-overlapping sets of variables.  Let a relaxed formula \Sup{S}{rlx}
be obtained from $S$ by dropping the clauses of $A$ i.e. $\Sup{S}{rlx}
= B$.  Let \pqe{W}{A}{H}{B} hold for a formula $H(Y)$ where $W = X
\cup Z$.  Then, $H$ is a boundary formula in terms of $Y$ for
relaxation \Sub{S}{rlx}.  That is from Proposition~\ref{prop:gen_bf}
it follows that a) \Impl{S}{H} and b) $H(\pnt{y})=0$ for every
assignment \pnt{y} to $Y$ that can be extended to satisfy \Sup{S}{rlx}
but not $S$.

Let $A \wedge B \equiv 0$ and $A \rightarrow H$ hold (the latter being
a stronger version of $S \rightarrow H$).  Then $H$ is an
interpolant~\cite{craig,pudlak,ken03} for implication
\Impl{A}{\overline{B}} (see Proposition~\ref{prop:interp_as_bf} of the
appendix). So an interpolant is a \ti{special case} of a boundary
formula.

Suppose that \Impl{A}{H} and $A \wedge B \not\equiv 0$ (and hence
\Nmpl{A}{\overline{B}}).  Then $H \wedge B \not\equiv 0$ and $H$ can
be viewed as an interpolant for the \ti{broken implication}
\Nmpl{A}{\overline{B}}. When \Impl{A}{\overline{B}} holds,
\Impl{H}{\overline{B}} gives a more abstract version of the
former. Similarly, if \Nmpl{A}{\overline{B}}, then
\Nmpl{H}{\overline{B}} is a more abstract version of the former.
Interpolants of broken implications can be used to generate \ti{short
  versions of counterexamples}.  A counterexample breaking
\Impl{H}{\overline{B}} can be extended to one breaking
\Impl{A}{\overline{B}} (see Proposition~\ref{prop:abs_cex} of the
appendix). So a counterexample for \Impl{H}{\overline{B}} is a short
version of that for \Impl{A}{\overline{B}}.

%
%
\subsection{Interpolation and LoR in the context of equivalence checking}
\label{ssec:int_bf_ec}

\setlength{\intextsep}{4pt}
\begin{wrapfigure}{l}{1.7in}
 \begin{center}
    \includegraphics[width=1.6in]{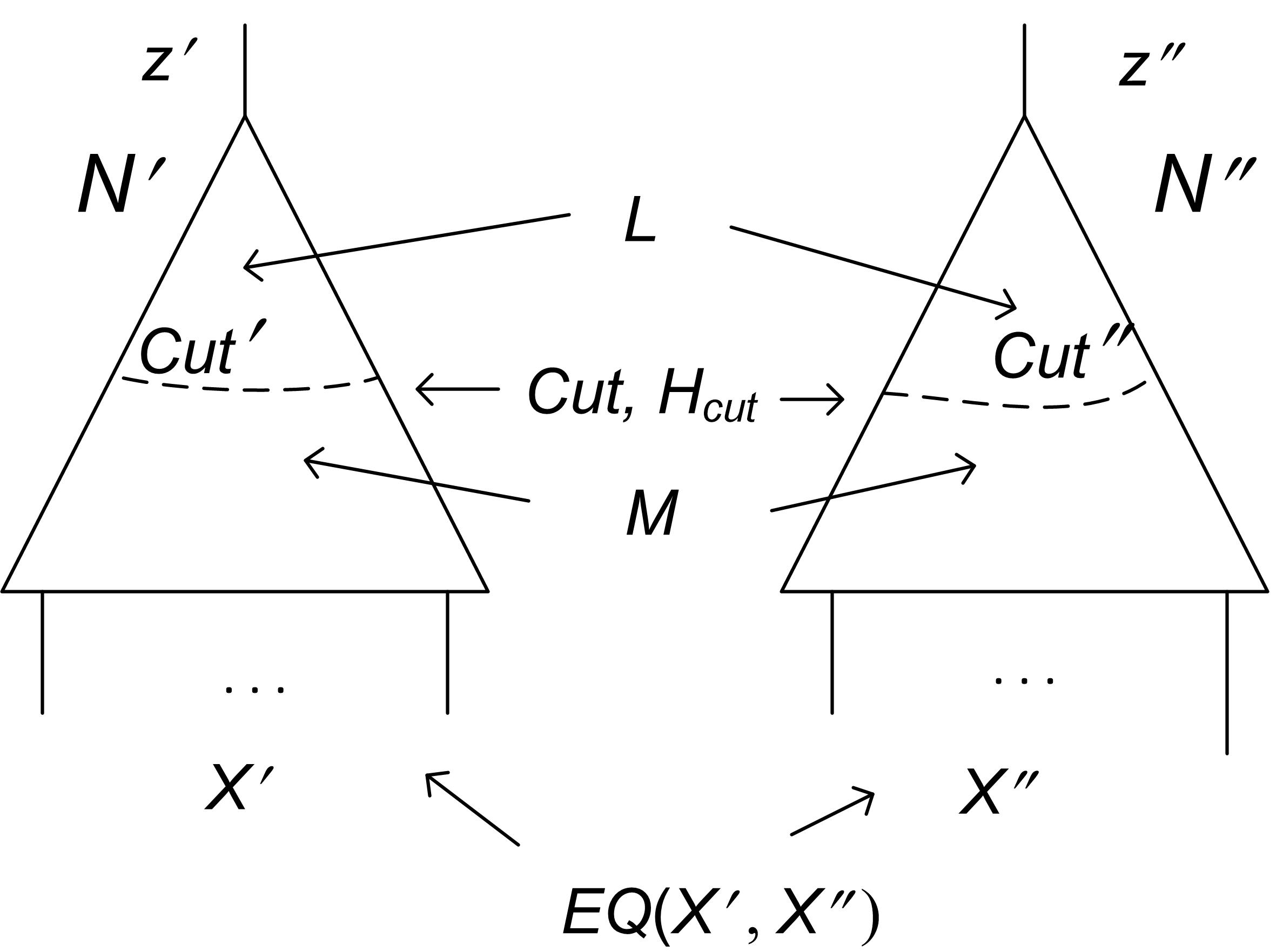}
  \end{center}
\vspace{-10pt}
\caption{Replacing/Separa-ting boundary formula \Sub{H}{cut}}
\label{fig:bbnd_form}
\end{wrapfigure}

In this subsection, we discuss the difference between boundary
formulas and interpolants in the context of EC. Let formulas $F_M$ and
$F_L$ specify the gates located below and above a cut as shown in
Fig.~\ref{fig:bbnd_form}. Then checking the equivalence of $N'$ and
$N''$ comes down to testing the satisfiability of formula $S$ equal to
$\mi{EQ}(X',X'') \wedge F_M \wedge F_L \wedge (z' \not\equiv z'')$.

Below, we contrast two types of relaxation of formula $S$ called
\ti{replacing} and \ti{separating} relaxation.  The former corresponds
to \ti{interpolation} while the latter is the relaxation we studied in
the previous sections. A \ti{replacing} relaxation of $S$ is to drop
the clauses of $\mi{EQ} \wedge F_M$. That is \Sup{S}{rlx} = $F_L
\wedge (z' \not\equiv z'')$.  Let \SUP{H}{cut}{r} be a boundary
formula computed for replacing relaxation. (Superscript $r$ stands for
``replacing''.)  That is \prob{W}{\mi{EQ} \wedge F_M \wedge F_L \wedge
  (z' \not\equiv z'')} $\equiv \SUP{H}{cut}{r} \wedge \prob{W}{F_L
  \wedge (z' \not\equiv z'')}$ where $W$ consists of all the variables
of $S$ but cut variables. Note that \SUP{H}{cut}{r} \ti{replaces} all
clauses depending on variables corresponding to gates below the cut,
hence the name replacing relaxation.  Let $A$ denote formula $\mi{EQ}
\wedge F_M$ and $B$ denote formula $F_L \wedge (z' \not\equiv z'')$.
From Proposition~\ref{prop:interp_as_bf} it follows that if $A
\rightarrow \SUP{H}{cut}{r}$ and $A \wedge B \equiv 0$ then
\SUP{H}{cut}{r} is an \ti{interpolant} of implication $A \rightarrow
\overline{B}$.  So an interpolant can be viewed as a boundary formula
for replacing relaxation.

A \ti{separating} relaxation of $S$ is to drop the clauses of
$\mi{EQ}$.  As we mentioned above, this kind of relaxation has been
the focus of the previous sections. Let \SUP{H}{cut}{s} denote a
boundary formula for separating relaxation. (Superscript $s$
stands for ``separating''.)  Formula \SUP{H}{cut}{s} satisfies
\PRob{W}{\mi{EQ} \wedge F_M \wedge F_L \wedge (z' \not\equiv z'')}
$\equiv \SUP{H}{cut}{s} \wedge$ \PRob{W}{F_M \wedge F_L \wedge (z'
  \not\equiv z'')}. Note that adding formula \SUP{H}{cut}{s} \ti{separates}
input variables $X'$ and $X''$ of $N'$ and $N''$ by making formula
$\mi{EQ(X',X'')}$ redundant, hence the name separating relaxation.  We
will refer to \SUP{H}{cut}{r} and \SUP{H}{cut}{s} as replacing and
separating boundary formulas respectively.

Let us assume for the sake of simplicity that a replacing boundary
formula \SUP{H}{cut}{r} is an interpolant i.e. it is implied by
$\mi{EQ} \wedge F_M$.  We will also assume that a separating boundary
formula \SUP{H}{cut}{s} satisfies the condition of
Proposition~\ref{prop:bf_by_pqe} and hence is implied by $\mi{EQ}
\wedge F_M$ as well. An obvious difference between \SUP{H}{cut}{r} and
\SUP{H}{cut}{s} is as follows. Adding \SUP{H}{cut}{s} to formula
\PRob{W}{\mi{EQ} \wedge F_M \wedge F_L \wedge (z' \not\equiv z'')}
makes redundant only a \ti{subset} of clauses that is made redundant
after adding \SUP{H}{cut}{r}.  The fact that adding \SUP{H}{cut}{r}
has to make redundant both clauses of $EQ$ and $F_M$ creates the
following problem with using interpolants for equivalence checking. On
the one hand, since \SUP{H}{cut}{r} is implied by $\mi{EQ} \wedge
F_M$, the former can be obtained by resolving clauses of the latter
i.e. without looking at the part of $N'$ and $N''$ above the cut. On
the other hand, proving that \SUP{H}{cut}{r} is indeed an interpolant,
in general, requires checking that $\SUP{H}{cut}{r} \wedge F_L \wedge
(z' \not\equiv z'') \equiv 0$ and hence needs the knowledge of the
part of $N'$ and $N''$ above the cut.

Informally, the problem above means that one cannot build a small
interpolant \SUP{H}{cut}{r} using only clauses of $\mi{EQ} \wedge
F_M$. By contrast, one can construct a small separating boundary
formula without any knowledge of formula $F_L$. Let us consider the
following simple example. Suppose that the cut of
Fig.~\ref{fig:bbnd_form} is an equivalence cut. That is for for every
cut point of $N'$ there is a functionally equivalent cut point of
$N''$ and vice versa. From Corollary~\ref{cor:eq_cut} it follows that
formula $\mi{EQ}(\mi{Cut'},\mi{Cut''})$ is a separating boundary
formula.  (Here $\mi{Cut'}$ and $\mi{Cut''}$ specify the cut points of
$N'$ and $''$ respectively.)  This fact can be established from
formula $\mi{EQ} \wedge F_M$ \ti{alone}.  However, whether
$\mi{EQ}(\mi{Cut'},\mi{Cut''})$ is an interpolant of implication $A
\rightarrow \overline{B}$ (where $A = \mi{EQ}(X',X'') \wedge F_M$ and
$B = F_L \wedge (z' \not\equiv z'')$) \ti{totally depends on formula}
$F_L$ i.e. on the part of $N'$ and $N''$ above the cut.

\vspace{-10pt}
\section{Conclusions}
\label{sec:conclusions}
We introduced a new framework for Equivalence Checking (EC) based on
Lo-gic Relaxation (LoR). The appeal of applying LoR to EC is
twofold. First, EC by LoR provides a powerful method for generating
proofs of equivalence by induction. Second, LoR gives a framework for
proving inequivalence without generating a counterexample.  The idea
of LoR is quite general and can be applied beyond EC.  LoR is enabled
by a technique called partial quantifier elimination and the
performance of the former strongly depends on that of the latter.  So
building efficient algorithms of partial quantifier elimination is of
great importance.

\section*{Acknowledgment}
I would like to thank Harsh Raju Chamarthi for reading the first
version of this paper.  My special thanks go to Mitesh Jain who has
read several versions of this paper and made detailed and valuable
comments.  This research was supported in part by NSF grants
CCF-1117184 and CCF-1319580.

\bibliographystyle{plain}
\bibliography{short_sat,local}
\vspace{15pt}
\appendix
\noindent{\large \tb{Appendix}}
\section{Proofs Of Propositions}
\label{app:proofs}
%
%
\setcounter{proposition}{0}
\begin{proposition}
Let $H(z',z'')$ be a formula such that $\prob{W}{\mi{EQ} \wedge
  \Sup{G}{rlx}} \equiv$ $H \wedge \prob{W}{\Sup{G}{rlx}}$ where $W =
X' \cup X'' \cup Y' \cup Y''$. Then formula $G \wedge (z' \not\equiv z'')$
is equisatisfiable with $H \wedge \Sup{G}{rlx} \wedge (z' \not\equiv z'')$.
\end{proposition}
\begin{proof}
A proof of this proposition follows from
Propositions~\ref{prop:bf_by_pqe}~and~\ref{prop:bf_prop} below.
Proposition~\ref{prop:bf_by_pqe} entails that $H(z',z'')$ is a
boundary formula. From Proposition~\ref{prop:bf_prop} it follows that
$G \wedge (z' \not\equiv z'')$ is equisatisfiable with $H \wedge
\Sup{G}{rlx} \wedge (z' \not\equiv z'')$.
\end{proof}

%
%
\begin{proposition}
Let \Sub{H}{cut} be a formula depending only on variables of a cut.
Let \Sub{H}{cut} satisfy \pqe{W}{\mi{EQ}}{\Sub{H}{cut}}{F_M}.  Here
$W$ is the set of variables of $F_M$ minus those of the cut.  Then
\Sub{H}{cut} is a boundary formula.
\end{proposition}
\begin{proof}
\pqe{W}{\mi{EQ}}{\Sub{H}{cut}}{F_M} entails $\mi{EQ} \wedge F_M
\rightarrow \Sub{H}{cut}$.  Let $L$ be the subcircuit consisting of the
gates of $N'$ and $N''$ located above the cut. Let $F_L$ be a formula
specifying $L$. Since $G = \mi{EQ} \wedge F_M \wedge F_L$, then
\Impl{G}{\Sub{H}{cut}} and so condition a) of
Definition~\ref{def:bnd_form} is met. Let us prove that condition b)
is met as well.  Let \pnt{q} be a cut assignment that can be extended
to satisfy \Sup{G}{rlx} but not $G$. This means that \pnt{q} cannot be
extended to an assignment \pnt{p} satisfying $\mi{EQ} \wedge F_M$
either. (Otherwise, one could easily extend \pnt{p} to an assignment
satisfying $\mi{EQ} \wedge F_M \wedge F_L$ and hence $G$ by using the
values of an execution trace computed for circuit $L$. This trace
describes computation of output values of $L$ when its input variables
i.e. the cut variables are assigned as in \pnt{q}.)  So
\prob{W}{\mi{EQ} \wedge F_M}=0 under assignment \pnt{q}.  This means
that $\Sub{H}{cut} \wedge \prob{W}{F_M}=0$ under assignment
\pnt{q}. Taking into account that \pnt{q} can be extended to an
assignment satisfying \Sup{G}{rlx} and hence $F_M$, one has to
conclude that $\Sub{H}{cut}(\pnt{q})=0$.
\end{proof}
%
%
%
\begin{proposition}
Let \Sub{H}{cut} be a boundary formula with respect to a cut.
Then  $G \wedge (z' \not\equiv z'')$ is
equisatisfiable with $\Sub{H}{cut}  \wedge \Sup{G}{rlx} \wedge (z' \not\equiv z'')$.
\end{proposition}
\begin{proof}
Let us show that the satisfiability of the left formula i.e.  $G
\wedge (z' \not\equiv z'')$ implies that of the right formula
i.e. $\Sub{H}{cut} \wedge \Sup{G}{rlx} \wedge (z' \not\equiv z'')$ and vice
versa.

\vspace{4pt} \tb{Left sat. $\rightarrow$ Right sat.}  Let \pnt{p} be
an assignment satisfying $G \wedge (z' \not\equiv z'')$. From
Definition~\ref{def:bnd_form} it follows that $G$ implies \Sub{H}{cut}
and so \Sub{H}{cut} is satisfied by \pnt{p}. Since \Sup{G}{rlx} is a
subformula of $G$, assignment \pnt{p} satisfies \Sup{G}{rlx} as
well. Hence \pnt{p} satisfies $\Sub{H}{cut} \wedge \Sup{G}{rlx} \wedge
(z' \not\equiv z'')$.

\tb{Right sat. $\rightarrow$ Left sat.} Let \pnt{p} be an assignment
satisfying $\Sub{H}{cut} \wedge \Sup{G}{rlx} \wedge (z' \not\equiv
z'')$. Let \pnt{q} be the subset of \pnt{p} consisting of the
assignments to the cut variables. Since \Sub{H}{cut}(\pnt{q})=1,
Definition~\ref{def:bnd_form} entails that \pnt{q} can be extended to
an assignment \pnt{p^*} satisfying formula $G$. Since the variables
assigned in \pnt{q} form a cut of circuits $N'$ and $N''$, the
consistent assignments to the variables of $N'$ and $N''$ located
above the cut are identical in \pnt{p} and \pnt{p^*}. This means that
\pnt{p^*} satisfies $(z' \not\equiv z'')$ and hence formula $G \wedge (z'
\not\equiv z'')$.
\end{proof}

%
%
\begin{proposition}
Let $\mi{Cut'},\mi{Cut}''$ specify the outputs of circuits $M'$ and
$M''$ of Fig.~\ref{fig:bnd_form} respectively. Assume that for every
variable $v'$ of $\mi{Cut'}$ there is a set $S(v') =
\s{v''_{i_1},\dots,v''_{i_k}}$ of variables of $\mi{Cut''}$ that have
the following property.  Knowing the values of variables of $S(v')$
produced in $N''$ under input \pnt{x} one can determine the value of
$v'$ of $N'$ under the same input \pnt{x}. We assume here that $S(v')$
has this property for every possible input \pnt{x}. Let
$\mi{Max}(S(v'))$ be the size of the largest $S(v')$ over variables of
$\mi{Cut'}$. Then there is a boundary formula \Sub{H}{cut} where every
clause has at most $\mi{Max}(S(v'))+1$ literals.
\end{proposition}
\begin{proof}
Let \pnt{q} be an assignment to the cut variables that can be extended
to satisfy formula \Sup{G}{rlx} but not formula $G$. To prove the
proposition at hand, one needs to show that there is a clause $C$
consisting of cut variables such that
\begin{itemize}
\item $C$ is implied by formula $G$
\item $C(\pnt{q}) = 0$
\item $C$ consists of at most $\mi{Max}(S(v'))$ literals
\end{itemize}
(Using clauses satisfying the three conditions above one can build 
a required boundary formula \Sub{H}{cut}.)

Let \pnt{p} be an assignment satisfying formula \Sup{G}{rlx} that is
obtained by extending \pnt{q}. Let \pnt{x'} and \pnt{x''} be the
assignments of \pnt{p} to variables of $X'$ and $X''$
respectively. Note that $\pnt{x'} \neq \pnt{x''}$ (otherwise \pnt{p}
would satisfy formula $G$ as well).  Cut assignment \pnt{q} can be
represented as (\pnt{q'},\pnt{q''}) where \pnt{q'} and \pnt{q''} are
assignments of \pnt{q} to $\mi{Cut'}$ and $\mi{Cut''}$ respectively.
Assignment \pnt{q'} (respectively \pnt{q''}) is produced by circuit
$M'$ (respectively $M''$) under input \pnt{x'} (respectively
\pnt{x''}).

Let $v'$ be a variable of $\mi{Cut}'$. The value of $v'$ is uniquely
specified by assignment \pnt{q''} to $S(v')$. So the value of every
variable of $\mi{Cut}'$ is specified by assignment \pnt{q''} to
$\mi{Cut}''$. Denote by \pnt{s'} the assignment to $\mi{Cut'}$
specified by \pnt{q''}. Let us show that $\pnt{s'} \neq \pnt{q'}$.
Assume the contrary i.e. \pnt{s'} = \pnt{q'} and show that then one
can extend \pnt{q} to an assignment \pnt{p^*} satisfying formula $G$
and so we have a contradiction.  Assignment \pnt{p^*} is constructed
as follows. The variables below the cut are assigned in \pnt{p^*} as
in the execution trace obtained by applying \pnt{x''} to $M'$ and
$M''$. Note that by assumption, applying input \pnt{x''} to $M'$ will
produce cut assignment \pnt{s'} equal to \pnt{q'}. The variables above
the cut are assigned in \pnt{p^*} as in \pnt{p}. Since \pnt{p}
satisfies \Sup{G}{rlx} and $X'$ and $X''$ are assigned the same input
\pnt{x''} in \pnt{p^*}, the latter satisfies $G$. Besides, the cut
assignment specified by \pnt{p^*} is \pnt{q} i.e. the same as the one
specified by \pnt{p}.

Since $\pnt{q'} \neq \pnt{s'}$, there is a variable $v'$ of
$\mi{Cut'}$ that is assigned in \pnt{q'} inconsistently with the
assignment of \pnt{q''} to the variables of $S(v')$. Let $C''$ be the
clause of variables of $S(v')$ falsified by \pnt{q''}. Let $l(v')$ be
the literal of $v'$ falsified by \pnt{q'}.  Then clause $l(v') \vee
C''$ is falsified by \pnt{q}.  The fact that assignment \pnt{q'} to
$S(v')$ determines the value of $v'$ means that clause $l(v') \vee
C''$ is implied by formula $F_{M'} \wedge F_{M''}$.  Hence $l(v') \vee
C''$ is implied by $G$. Finally, the number of literals in $l(v') \vee
C''$ is $|S(v')| + 1$. So clause $l(v') \vee C''$ satisfies the three
conditions above.
\end{proof}
%
%
\begin{proposition}
Let $W_i$ where $i > 0$ be the set of variables of $F_{M_i}$ minus
those of $\mi{Cut}_i$.  Let $H_{i-1}$ where $i > 1$ be a boundary
formula such that $\prob{W_{i-1}}{H_0 \wedge F_{M_{i-1}}} \equiv
H_{i-1} \wedge \exists{W_{i-1}} [F_{M_{i-1}}]$.  Let
\pqe{W_i}{H_{i-1}}{H_i}{F_{M_i}} hold. Then \PRob{W_i}{H_0 \wedge
  F_{M_i}} $\equiv$ $H_i \wedge \prob{W_i}{F_{M_i}}$ holds. (So $H_i$
is a boundary formula due to Proposition~\ref{prop:bf_by_pqe}.)
\end{proposition}
\begin{proof}
Let $\phi$ denote formula $\prob{W_i}{H_0 \wedge F_{M_i}}$. Let
$F_{i-1,i}$ be the set of clauses equal to $F_{M_i} \setminus
F_{M_{i-1}}$. Formula $\phi$ can be represented as
\Prob{W_{i-1}}{W_{i-i,i}}{H_0 \wedge F_{M_{i-1}}\wedge F_{i-1,i}}
where $W_{i-1,i} = W_i \setminus W_{i-1}$. Taking into account that
formula $F_{i-1,i}$ does not depend on variables of $W_{i-1}$, one can
rewrite formula $\phi$ as $\exists{W_{i-1,i}}[F_{i-1,i} \wedge
  \prob{W_{i-1}}{H_0 \wedge F_{M_{i-1}}}]$.  Using the assumption
imposed on $H_{i-1}$ by the proposition at hand, one can transform
formula $\phi$ into $\exists{W_{i-1,i}}[F_{i-1,i} \wedge H_{i-1}
  \wedge \prob{W_{i-1}}{F_{M_{i-1}}}]$.  After putting $F_{i-1,i}$ and
$H_{i-1}$ back under the scope of quantifiers, $\phi$ becomes equal to
$\exists{W_{i-1,i}}[ \prob{W_{i-1}}{H_{i-1} \wedge F_{M_{i-1}} \wedge
    F_{i-1,i}}]$ and hence to $\exists{W_i}[H_{i-1} \wedge F_{M_i}]$.
Since \pqe{W_i}{H_{i-1}}{H_i}{F_{M_i}} holds we get that the original
formula $\phi$ equal to \prob{W_i}{H_0 \wedge F_{M_i}} is logically
equivalent to $H_i \wedge \prob{W_i}{F_{M_i}}$.
\end{proof}
%
%
\begin{proposition}
\label{prop:gen_method}
Let $S(X,Z)$, $\Sup{S}{rlx}(X,Z)$, $E(X,Z)$ and $H(Z)$ be Boolean
formulas where $X,Z$ are non-overlapping sets of variables.  Let $S =
E \wedge \Sup{S}{rlx}$ and $\prob{X}{E \wedge \Sup{S}{rlx}} \equiv H
\wedge \prob{X}{\Sup{S}{rlx}}$ hold.  Then $S$ is equisatisfiable with
$H \wedge \Sup{S}{rlx}$.
\end{proposition}
\begin{proof}
By assumptions of the proposition, $\prob{X}{S} \equiv \prob{X}{H
  \wedge \Sup{S}{rlx}}$.  So if formula $S$ is satisfiable, there is
an assignment \pnt{z} to the variables of $Z$ for which \prob{X}{S}
evaluates to 1. Since formula \prob{X}{H \wedge \Sup{S}{rlx}} also
evaluates to 1 for \pnt{z}, formula $H \wedge \Sup{S}{rlx}$ is
satisfiable too.  Similarly, one can show that the satisfiability of
\prob{X}{H \wedge \Sup{S}{rlx}} means that that $S$ is satisfiable
too.
\end{proof}
%
%
\begin{proposition}
\label{prop:gen_bf}
Let formula $S(X,Z)$ be represented as $E(X,Z) \wedge
\Sup{S}{rlx}(X,Z)$ where $X,Z$ are non-overlapping sets of Boolean
variables.  Let \pqe{X}{E}{H}{\Sup{S}{rlx}} hold for a formula $H(Z)$.
Then $H$ is a boundary formula in terms of $Z$ for relaxation
\Sup{S}{rlx} (see Definition~\ref{def:bnd_form}).  That is
\begin{enumerate}[a)]
\vspace{-2pt}
\item $S \rightarrow H$ and
\item for every assignment \pnt{z} to $Z$ that can be extended to satisfy \Sup{S}{rlx} but
not $S$, the value of $H(\pnt{z})$ is 0.
\end{enumerate}
\end{proposition}
\begin{proof}
\pqe{X}{E}{H}{\Sup{S}{rlx}} entails $E \wedge \Sup{S}{rlx} \rightarrow
H$. So condition a) is met. Let us show that condition b) holds as
well. Let \pnt{z} be an assignment to $Z$ that can be extended to
satisfy \Sup{S}{rlx} but not $S$. This means that \prob{X}{E \wedge
  \Sup{S}{rlx}} and \prob{X}{\Sup{S}{rlx}} evaluate to 0 and 1
respectively under assignment \pnt{z}. Hence $H(\pnt{z})$ has to be
equal to 0 to preserve \pqe{X}{E}{H}{\Sup{S}{rlx}}.
\end{proof}
%
%
\begin{proposition}
\label{prop:interp_as_bf}
Let $A(X,Y)$ and $B(Y,Z)$ be formulas where $X,Y,Z$ are
non-overlapping sets of variables. Let $A \wedge B \equiv 0$. Formula
$H(Y)$ is an interpolant of implication \Impl{A}{\overline{B}} iff
\Impl{A}{H} and \pqe{W}{A}{H}{B} where $W = X \cup Z$.
\end{proposition}
\begin{proof}
\ti{If part}.  Suppose that \Impl{A}{H} holds and \pqe{W}{A}{H}{B}.
Since \Impl{A}{\overline{B}} holds, then $A \wedge B \equiv 0$ and so
$H \wedge B \equiv 0$. Hence \Impl{H}{\overline{B}} and $H$ is an
interpolant of implication \Impl{A}{\overline{B}}.

\ti{Only if part}. Suppose that $H$ is an interpolant and so
\Impl{A}{H} and \Impl{H}{\overline{B}} hold.  Assume that $\prob{W}{A
  \wedge B} \not\equiv H \wedge \prob{W}{B}$. Since
\Impl{H}{\overline{B}} and hence $H \wedge B \equiv 0$, this means
that $A \wedge B \not\equiv 0$. So we have a contradiction.
\end{proof}

%
%
\begin{proposition}
\label{prop:abs_cex}
Let $A(X,Y) \wedge B(Y,Z) \not\equiv 0$ where $X,Y,Z$ are
non-overlapping sets of variables. Let $H(Y)$ be a formula such that
\Pqe{W}{A}{H}{B} where $W = X \cup Z$.  Let \pnt{y} and \pnt{z} be
assignments to $Y$ and $Z$ respectively such that (\pnt{y},\pnt{z})
satisfies $H \wedge B$. Then (\pnt{y},\pnt{z}) can be extended to an
assignment satisfying $A \wedge B$.
\end{proposition}
\begin{proof}
The fact that \Pqe{W}{A}{H}{B} holds and $H \wedge B$ is satisfied by
(\pnt{y},\pnt{z}) means that \pnt{y} can be extended to an assignment
(\pnt{x^*},\pnt{y},\pnt{z^*}) satisfying $A \wedge B$. Then assignment
(\pnt{x^*},\pnt{y},\pnt{z}) satisfies $A \wedge B$ as well. Indeed,
(\pnt{x^*},\pnt{y}) satisfies $A$ and (\pnt{y},\pnt{z}) does $B$.
\end{proof}

\section{Algorithm For Partial Quantifier Elimination}
\label{app:pqe}
In this section, we discuss Partial Quantifier Elimination (PQE) in
more detail. In Subsection~\ref{app_s:pqe_alg}, we give a high-level
description of a PQE-solver. This PQE-solver is based on the machinery
of Dependency sequents (D-sequents) that we recall in
Subsection~\ref{app_s:dep_seqs}.

%
%
\subsection{A PQE solver}
\label{app_s:pqe_alg}

In this subsection, we describe our algorithm for PQE introduced
in~\cite{hvc-14} in 2014.  We will use the same name for this
algorithm as in Section~\ref{sec:experiments}, i.e. PQE-14.  Let
$A(V,W), B(V,W)$ be Boolean formulas where $V$,$W$ are non-overlapping
sets of variables. As we mentioned in Subsection~\ref{subsec:pqe_def},
the PQE problem is to find formula $A^*(V)$ such that
\pqe{W}{A}{A^*}{B}. We will refer to a clause containing a variable of
$W$ as a \pnt{W}\tb{-clause}. PQE-14 is based on the three ideas
below.

\ti{First}, finding formula $A^*$ comes down to generation of clauses
depending only on variables of $V$ that make the $W$-clauses of $A$
redundant in $A^* \wedge \prob{W}{A \wedge B}$.  \ti{Second}, the
clauses of $A^*$ can be derived by resolving clauses of $A \wedge
B$. The intermediate resolvents that are $W$-clauses need to be proved
redundant along with the original $W$-clauses of $A$.  However, since
formula $B$ remains quantified, there is no need to prove redundancy
of $W$-clauses of $B$ or $W$-clauses obtained by resolving only
clauses of $B$.

\ti{Third}, since proving redundancy of a clause is a hard problem it
makes sense to partition this problem into simpler subproblems. To
this end, PQE-14 employs branching. After proving redundancy of
required clauses in subspaces, the results of branches are merged. The
advantage of branching is that for every $W$-clause $C$ one can always
reach a subspace where $C$ can be trivially proved redundant. Namely,
$C$ is trivially redundant in the current subspace if a) $C$ is
satisfied in the current branch; b) $C$ is implied by some other
clause; c) there is an unassigned variable $y$ of $C$ where $y \in W$,
such that $C$ cannot be resolved on $y$ with other clauses that are
not satisfied or proved redundant yet.

%
%
\subsection{Dependency sequents}
\label{app_s:dep_seqs}
PQE-14 branches on variables of $V \cup W$ until the $W$-clauses that
are descendants of $W$-clauses of $A$ are proved redundant in the
current subspace. To keep track of conditions under which a $W$-clause
becomes redundant in a subspace, \mbox{PQE-14} uses the machinery of
Dependency sequents (D-sequents) developed
in~\cite{fmcad12,fmcad13}. A \tb{D-sequent} is a record of the form
\Dss{W}{A \wedge B}{q}{\s{C}}. It states that clause $C$ is redundant
in formula \prob{W}{A \wedge B} in subspace \pnt{q}. Here \pnt{q} is
an assignment to variables of $V \wedge W$ and $A$ is the \ti{current}
formula that consists of the initial clauses of $A$ and the resolvent
clauses.  When a $W$-clause $C$ is proved redundant in a subspace,
this fact is recorded as a D-sequent. If $\pnt{q}=\emptyset$, the
D-sequent is called \ti{unconditional}. Derivation of such a D-sequent
means that clause $C$ is redundant in the current formula \prob{W}{A
  \wedge B} in the entire space.

The objective of PQE-14 is to derive unconditional D-sequents for all
$W$-clauses of $A$ and their descendants that are $W$-clauses.  A new
D-sequent can be obtained from two parent D-sequents by a
resolution-like operation on a variable $y$. This operation is called
\ti{join}. When PQE-14 merges the results of branching on variable $y$
it joins D-sequents obtained in branches $y=0$ and $y=1$ at variable
$y$.  So the resulting D-sequents do not depend on $y$.  If formula $A
\wedge B$ is unsatisfiable in both branches, a new clause $C$ is added
to formula $A$. Clause $C$ is obtained by resolving a clause falsified
in subspace $y=0$ with a clause falsified in subspace $y=1$ on
$y$. Adding $C$ makes all $W$-clauses redundant in the current
subspace.  By the time PQE-14 backtracks to the root of the search
tree, it has derived unconditional D-sequents for all $W$-clauses of
the current formula $A$.

Algorithms based on D-sequents (including PQE solving) is work in
progress. So they still lack some important techniques like D-sequent
re-using.  In the current algorithms based on D-sequents, the parent
D-sequents are discarded as soon as they produce a new D-sequent by
the join operation. Although D-sequent re-using promises to be as
powerful as re-using learned clauses in SAT-solving, it requires more
sophisticated bookkeeping and so is not implemented yet~\cite{hvc-14}.

%
%
\section{Generation Of Cuts That Do Not Overlap}
\label{app:cuts_levels}
An important part of \EC described in Section~\ref{sec:ec_by_lor} is
to build non-overlapping cuts. These cuts are used to generate a
sequence of boundary formulas converging to an output boundary
formula. As we mentioned there, the presence of non-local connections
makes it hard to find cuts that do not overlap. In this section, we
consider this issue in more detail.  First, we give the necessary
definitions and describe the problem. Then we explain how one can get
rid of non-local connections by buffer insertion.

Let $M$ be a multi-output circuit. The \ti{length} of a path
from an output of a gate to an input of another gate is measured by
the number of gates on this path.The \ti{topological level} of a gate
$g$ is the longest path from an input of $M$ to $g$. We treat the
inputs of $M$ as special gates that are not fed by other gates. We
will denote the topological level of gate $g$ as \tl{g}. It can be
computed recursively as follows. If $g$ is an input, then \tl{g} =
0. Otherwise, \tl{g} is equal to the maximum topological level among
the gates feeding $g$ plus 1.

\setlength{\intextsep}{4pt}
\begin{wrapfigure}{l}{1in}
 \begin{center}
    \includegraphics[height=1.2in,width=0.5in]{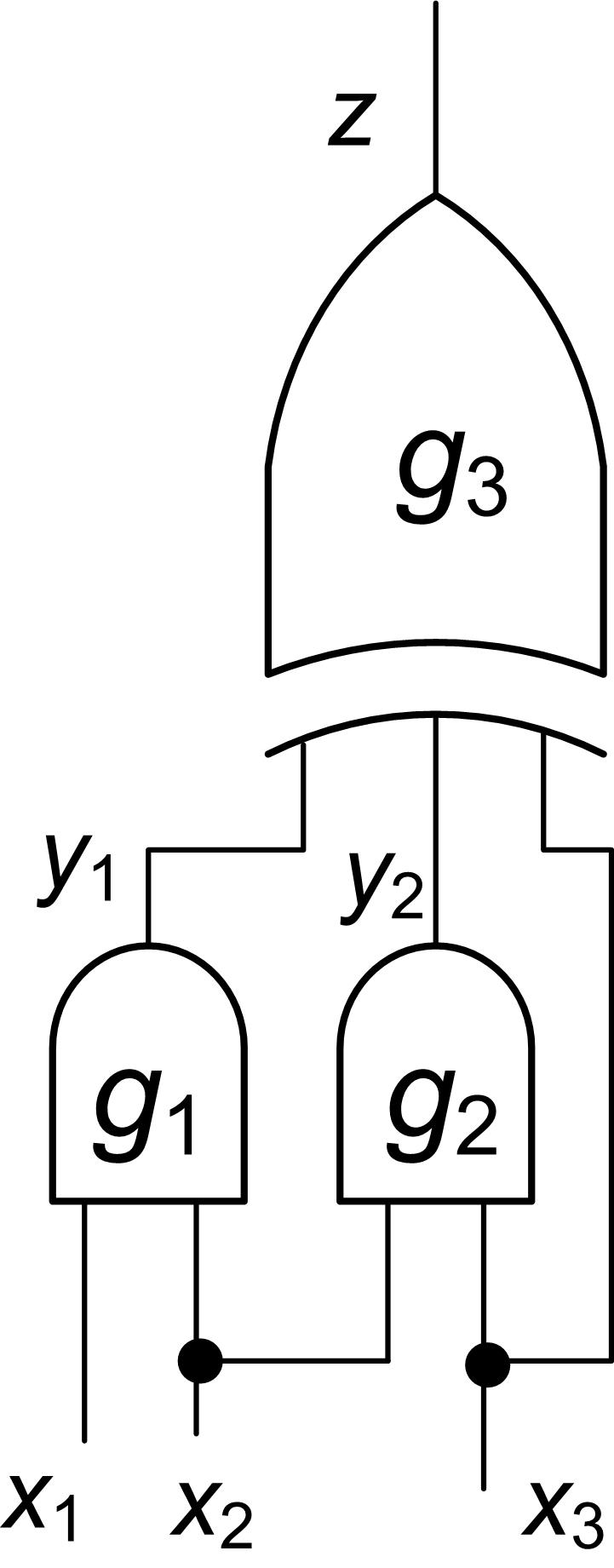}
  \end{center}
\vspace{-10pt}
\caption{A circuit}
\vspace{15pt}
\label{fig:top_lvls}
\end{wrapfigure}

We will call gates $g_i$ and $g_j$ \ti{topologically independent} if
there is no path from an input to an output of $M$ going through both
these gates. For instance, gates $g_1$ and $g_2$ in
Fig.~\ref{fig:top_lvls} are topologically independent.  We will call a
set $S$ of gates a \tb{cut}, if every path from an input to an output
of $M$ goes through a gate of $S$.  A cut $S$ is \ti{minimal}, if for
every gate $g \in S$, set \mbox{$S \setminus \s{g}$} is not a cut. \EC
employs only minimal cuts. In this section, we use the notion of a
gate and the variable specifying its output interchangeably. For
example, the topological level of a variable $v$ specifying the output
of gate $g$ (denoted as \tl{v}) is equal to \tl{g}.

If gate $g_i$ of $M$ feeds gate $g_j$ and $\tl{g_j} > \tl{g_i}+1$,
then $g_i$ and $g_j$ are said to have a \tb{non-local connection}.
Non-local connections make topologically dependent gates appear on the
same cut. Consider the circuit of Fig~\ref{fig:top_lvls}. The input
gate $x_3$ feeds gates $g_2$ and $g_3$.  Since \tl{x_3}=0 and
\tl{g_3}=2, the connection between $x_3$ and $g_3$ is non-local.  This
leads to appearance of cut \s{y_1,y_2,x_3} where variables $y_2$ and
$x_3$ are topologically dependent.  If gate $g_i$ feeds gate $g_j$ and
this connection is non-local, gate $g_i$ appears in every cut that
separates $g_i$ and $g_j$ and does not include $g_j$.  So the presence
of a large number of non-local connections leads to the heavy
overlapping of cuts.

There are a few techniques for dealing with non-local connections of
$N'$ and $N''$ in the context of EC by LoR. The simplest one is to
insert buffers. A buffer is a single-input and single-output gate that
copies its input to the output. Let $g_i$ and $g_j$ be gates of $N'$
such that a) $g_i$ feeds $g_j$ and b) $\tl{g_j} > \tl{g_i}+1$. By
inserting $\tl{g_j}-\tl{g_i}-1$ buffers between $g_i$ to $g_j$, this
non-local connection is replaced with $\tl{g_j}-\tl{g_i}$ local
connections.

%
%
\section{Version of \EC Used In Experiments}
\label{app:version}
In the experiments of Subsection~\ref{subsec:ec_lor}, we used a
version of \EC that was modified in comparison to the description
given in Fig.~\ref{fig:ec_by_lor}.  We will refer to this version as
\ecm. In this section, we describe \ECM in more detail.

Boundary formula $H_i$ was computed in \ECM as follows. If there was a
variable specifying the output of a cut gate $g'$ that was not present
in a clause of $H_i$, \ECM called the procedure below. That procedure
generated short clauses relating the output variable of $g'$ and those
of its ``relatives'' from $N''$. This way, \ECM avoided running a cut
termination check before every variable of \mbox{$i$-th} cut was
present in a clause of $H_i$.

Clauses of $H_i$ constraining variable of $g'$ were generated as
follows.  First, \ECM identified the relatives $g''_1,\dots,g''_m$ of
gate $g'$ in $N''$. A gate $g''_j$ was considered a relative of $g'$
if there was a clause of formula $H_{i-1}$ relating input variables of
$g'$ and $g''_j$.  Finally, a set of clauses $A^*$ relating the output
variable of gate $g'$ and those of its relatives was generated and
added to formula $H_i$. The clauses of $A^*$ were obtained by taking
formula $A$ out of the scope of quantifiers in \prob{W}{A \wedge B}.
Here $A$ is the set of clauses of formulas $H_{i-1}$ containing the
input variables of gate $g'$ and its relatives. Formula $B$ contains
the clauses specifying gate $g'$ and its relatives. Set $W$ consists
of the variables of $A \wedge B$ minus output variables of $g'$ and
its relatives.

Another modification of \ECM was that boundary formulas were computed
approximately. In line 11 of Fig~\ref{fig:ec_by_lor}, formula
$F_{M_i}$ specifying the gates located between inputs of $N'$ and
$N''$ and $i$-th cut is used to compute a new clause of $H_i$.  In
\ECM only the subset of clauses of $F_{M_i}$ specifying the gates
located between $(i-1)$-th and $i$-th cuts was used when computing
$H_i$.

\section{Computing Boundary Formulas Efficiently}
\label{app:eff_comp_bfs}
Computation of a boundary formula is based on PQE solving. In turn, a
PQE-solver is based on derivation of D-sequents (see
Subsection~\ref{app_s:dep_seqs} of the appendix).  As we showed in
Subsection~\ref{ssec:ind_comp_bf}, boundary formula $H_i$ is obtained
by taking $H_{i-1}$ out of the scope of quantifiers in formula
\prob{W_i}{H_{i-1} \wedge F_{M_i}}.  Here $F_{M_i}$ specifies the
gates located between inputs of circuits $N',N''$ and $i$-th cut and
$W_i$ is the set of variables of $F_{M_i}$ minus those of the $i$-th
cut. Since the size of formula $F_{M_i}$ grows with $i$, a PQE-solver
that computes $H_i$ \ti{precisely} must have high scalability.  PQE-14
(see Section~\ref{app:pqe}) does not scale well yet because it does
not re-use D-sequents.  In this section, we argue that once D-sequent
re-using is implemented, efficient computation of boundary formulas
becomes quite possible.

Consider the scalability problem in more detail. Formula $H_i$ is
obtained by generating a set of clauses that make the clauses of
$H_{i-1}$ redundant. Let $C \in H_{i-1}$ be a clause whose redundancy
one needs to prove. PQE-14 is a branching algorithm. Clause $C$ is
trivially redundant in every subspace where $C$ is satisfied. Proving
redundancy is non-trivial only in the subspace where $C$ is
falsified. To prove $C$ redundant in such a subspace, PQE-14 uses the
machinery of local proofs of redundancy described below. (For the sake
of simplicity we did not mention this aspect of PQE-14 in
Section~\ref{app:pqe}.)

Suppose clause $C$ above contains the positive literal of variable
$v$.  Suppose, in the current branch, literal $v$ is unassigned and
all the other literals of $C$ are falsified. So $C$ is currently a
unit clause. Then PQE-14 marks all the clauses containing literal
$\overline{v}$ that can be resolved with $C$ as ones that have to be
proved redundant in branch $v=1$ i.e. \ti{locally}. This is done even
for clauses of $F_{M_i}$ (that do not have to be proved redundant
\ti{globally} because $F_{M_i}$ remains quantified). The obligation to
prove redundancy of clauses with literal $\overline{v}$ is made to
prove redundancy of $C$ in the branch where $v=0$ and $C$ is
falsified. This obligation is canceled immediately after PQE-14
backtracks to the node $w$ of the search tree that precedes node $v$.
When exploring branch $v=1$ one of the two alternatives occurs. If
formula is UNSAT in this branch, a new clause is generated that
subsumes $C$ in node $w$. Adding this clause to \prob{W_i}{H_{i-1}
  \wedge F_{M_i}} makes $C$ redundant in node $w$.  Otherwise, clauses
with literal $\overline{v}$ are proved redundant in node $w$ and so
$C$ is redundant in node $w$ because it cannot be resolved on $v$ in
the current subspace. (This also means that $C$ is redundant in branch
$v=1$ where $C$ is falsified.)

To prove that a clause $B$ with literal $\overline{v}$ is redundant in
branch $v=1$ one may need to make obligations to prove redundancy of
some other clauses that can be resolved with $B$ on one of its
variables and so on. So proving redundancy of one clause $C$ makes
PQE-14 prove local redundancy of many clauses. Currently PQE-14
discards a D-sequent as soon as it is joined at a branching variable
of the search tree (with some other D-sequent). Moreover, the
D-sequent of a clause that one needs to prove only locally is
discarded after the obligation to prove redundancy of this clause is
canceled.  This cripples the scalability of PQE-14 because one has to
reproduce D-sequents seen before over and over again.  As the size of
formula $F_{M_i}$ grows, more and more clauses need to be proved
redundant locally and the size of the search tree blows up.

Re-using D-sequents should lead to drastic reduction of the search
tree size for two reasons.  First, when proving redundancy of clause
$C$ one can immediately discard every clause whose D-sequent states
the redundancy of this clause in the current subspace.  Second, one
can re-use D-sequents of clauses of $F_{M_j}$, $j < i$ that were
derived when building formula $H_j$. Informally, D-sequent re-using
should boost the performance of a PQE algorithm like re-using learned
clauses boosts that of a SAT-solver.

\end{document}